\documentclass[a4paper,UKenglish,numberwithinsect]{lipics-v2018-modified} \usepackage[full]{optional}

\usepackage[obeyDraft,textsize=scriptsize]{todonotes} \newcommand{\todoi}[1]{\todo[inline]{#1}}

\usepackage{amsmath}
\usepackage{amssymb}
\usepackage{graphicx}

\usepackage[nomessages]{fp}

\usepackage{algorithm}
\usepackage[noend]{algpseudocode}

\usepackage{float}

\usepackage{comment}

\usepackage{subcaption} 
\usepackage{caption}
\usepackage{cite}

\usepackage{color}
\usepackage{hyperref}
\usepackage{gensymb}

\makeatletter
\def\blfootnote{\xdef\@thefnmark{}\@footnotetext}
\def\BState{\State\hskip-\ALG@thistlm}

\makeatother

\newcommand{\removed}[1]{}

\newcommand{\ra}{\rightarrow}

\newcommand{\E}{\textup{E}}

\definecolor{DarkRed}{RGB}{182,11,1}

\usepackage{mathtools}

\newcommand{\ceil}[1]{\ensuremath{\left\lceil #1 \right\rceil}}
\newcommand{\floor}[1]{\ensuremath{\left\lfloor #1 \right\rfloor}}

\algnewcommand{\IfThenElse}[3]{
	\State \algorithmicif\ #1\ \algorithmicthen\ #2\ \algorithmicelse\ #3}

\title{Exact size counting in uniform population protocols in nearly logarithmic time} 

	\author
	{David Doty}
	{Department of Computer Science, University of California, Davis}
	{doty@ucdavis.edu}
	{}
	{\opt{full}{Supported by NSF grant CCF-1619343.}}
    
    \author
    {Mahsa Eftekhari}
    {Department of Computer Science, University of California, Davis}
    {mhseftekhari@ucdavis.edu}
    {}
    {\opt{full}{Supported by NSF grant CCF-1619343.}}
    
    \author
    {Othon Michail}
    {Department of Computer Science, University of Liverpool, UK}
    {Othon.Michail@liverpool.ac.uk}
    {}
    {\opt{full}{Supported by EEE/CS initiative NeST.}}
    
    \author
    {Paul G. Spirakis}
    {Department of Computer Science, University of Liverpool, UK and Computer Technology Institute \& Press ``Diophantus'' (CTI), Patras, Greece}
    {P.Spirakis@liverpool.ac.uk}
    {}
    {\opt{full}{Supported by EEE/CS initiative NeST.}}
    
    \author
    {Michail Theofilatos}
    {Department of Computer Science, University of Liverpool, UK}
    {michail.theofilatos@liverpool.ac.uk}
    {}
    {\opt{full}{Supported by EEE/CS initiative NeST, Leverhulme Research Centre for Functional Materials Design.}}

\authorrunning{D. Doty, M. Eftekhari, O. Michail, P.\,G. Spirakis, and M. Theofilatos}

\Copyright{David Doty, Mahsa Eftekhari, Othon Michail, Paul G. Spirakis, and Michail Theofilatos}

\opt{sub}{
    \subjclass{
        CCS 
        $\rightarrow$ 
        Theory of computation 
        $\rightarrow$ 
        Design and analysis of algorithms 
        $\rightarrow$ 
        Distributed algorithms
    }
}
\opt{full}{\subjclass{}}

\opt{sub}{}
\opt{full}{\keywords{}}

\opt{sub}{
    \funding{
    DD, ME: NSF grant CCF-1619343. 
    OM, PS, MT: EEE/CS initiative NeST. 
    MT: Leverhulme Research Centre for Functional Materials Design.}
}

\opt{sub}{
    \EventEditors{Ulrich Schmid}
    \EventNoEds{1}
    \EventLongTitle{International Symposium on Distributed Computing}
    \EventShortTitle{DISC 2018}
    \EventAcronym{DISC}
    \EventYear{2018}
    \EventDate{Oct. 15 - Oct. 18, 2018}
    \EventLocation{New Orleans, LA}
    \EventLogo{}
    \SeriesVolume{1}
    \ArticleNo{1}
}
\opt{full}{
    \EventEditors{}
    \EventNoEds{0}
    \EventLongTitle{}
    \EventShortTitle{}
    \EventAcronym{}
    \EventYear{}
    \EventDate{}
    \EventLocation{}
    \EventLogo{}
    \SeriesVolume{}
    \ArticleNo{}
    \nolinenumbers 
}

\begin{document}

\maketitle

\def\Cconstant{6}
\FPeval{\LCconstant}{clip(\Cconstant*2)}
\FPeval{\LCsubprotocolStateCountConstant}{clip(\Cconstant+\LCconstant)}
\FPeval{\Mconstant}{clip(\Cconstant*3)}
\newcommand{\aveconstant}{\Mconstant}
\newcommand{\countconstant}{\Cconstant}

\FPeval{\stateCountExponent}{clip(\Cconstant+\LCconstant+\Mconstant+\aveconstant+\countconstant)}

\def\CconstantMin{3}
\FPeval{\LCconstantMin}{clip(\CconstantMin*2)}
\FPeval{\LCsubprotocolStateCountConstantMin}{clip(\CconstantMin+\LCconstantMin)}
\FPeval{\MconstantMin}{clip(\CconstantMin*3)}
\newcommand{\aveconstantMin}{\MconstantMin}
\newcommand{\countconstantMin}{\CconstantMin}

\FPeval{\stateCountExponentMin}{clip(\CconstantMin+\LCconstantMin+\MconstantMin+\aveconstantMin+\countconstantMin)}

\newcommand{\stateCountConstantLog}{\ensuremath{15}}
\newcommand{\stateCountConstant}{\ensuremath{2^{\stateCountConstantLog}}}
\newcommand{\lowerBoundOnM}{\ensuremath{3n^3}}

\def\exactCountingConvergenceTimeConstant{6}
\FPeval{\expectedTimeConstant}{clip(\exactCountingConvergenceTimeConstant+1)}

\def\exactCountingConvergenceTimeProbErrNumerator{10 + 5 \log \log n}
\def\exactCountingConvergenceTimeProbErr{\frac{\exactCountingConvergenceTimeProbErrNumerator}{n}}

\newcommand{\epidemicExpectedTimeConstant}{4}

\newcommand{\C}{\texttt{C}}
\newcommand{\LC}{\texttt{LC}}
\newcommand{\isLeader}{\texttt{isLeader}}
\newcommand{\M}{\texttt{M}}
\newcommand{\ave}{\texttt{ave}}
\newcommand{\thecount}{\texttt{count}}
\newcommand{\phase}{\texttt{phase}}
\newcommand{\maxPhase}{\texttt{MaxPhase}}

\FPeval{\maxPhaseValue}{1184}
\FPeval{\betaUValue}{clip(12*\maxPhaseValue)}
\newcommand{\aveTimeConstant}{37}

\newcommand{\exactCounting}{\textsc{ExactCounting}}
\newcommand{\uniqueID}{\textsc{UniqueID}}
\newcommand{\electLeader}{\textsc{ElectLeader}}
\newcommand{\averaging}{\textsc{Averaging}}
\newcommand{\timer}{\textsc{Timer}}
\newcommand{\extendCode}{\textsc{ExtendCode}}
\newcommand{\setNewLeaderCode}{\textsc{SetNewLeaderCode}}
\newcommand{\randbits}{\textsc{RandBits}}
\newcommand{\append}{\textsc{Append}}
\newcommand{\false}{\textsc{False}}
\newcommand{\true}{\textsc{True}}

\newcommand{\rec}{\ensuremath{\textrm{rec}}}
\newcommand{\sen}{\ensuremath{\textrm{sen}}}

\algnewcommand{\LeftComment}[1]{\Statex \(\triangleright\) #1}

\newcommand{\round}[1]{\ensuremath{\textsc{Round}_{#1}}}

\todoi{Here is the bit requirement for each variable in \exactCounting:

$|\C| = \Cconstant \log n$

$|\LC| = \LCconstant \log n$

$\isLeader$ = 1 bit

For integers $x$ below, $|x|$ means length of binary expansion of $x$.

$|\M| = \log 3 + \Mconstant \log n < 2 +  \Mconstant \log n$

$|\ave| = 2 + \aveconstant \log n$

$|\thecount| = \countconstant \log n$

$|\phase| = \log \maxPhaseValue < 12$
}

\todoi{DD: I worked on this on the place and made changes. Just in case there are other changes I renamed the old main.tex file to main-bak.tex as of about 10:30pm PST.}

\begin{abstract}
We study population protocols:
networks of anonymous agents whose pairwise interactions are chosen uniformly at random.
The \emph{size counting problem} is that of calculating the exact number $n$ of agents in the population,
assuming no leader (each agent starts in the same state).
We give the first protocol that solves this problem in sublinear time.

The protocol converges in
$O(\log n \log \log n)$ time
and uses
$O(n^{\stateCountExponent})$ states
($O(1) + \stateCountExponent \log n$ bits of memory per agent)
with probability $1-O(\frac{\log \log n}{n})$.
The time to converge is also $O(\log n \log \log n)$ in expectation.
Crucially, unlike most published protocols with $\omega(1)$ states,
our protocol is \emph{uniform}:
it uses the same transition algorithm for any population size,
so does not need an estimate of the population size to be embedded into the algorithm.
A sub-protocol is the first uniform sublinear-time leader election population protocol,
taking $O(\log n \log \log n)$ time and $O(n^{\LCsubprotocolStateCountConstant})$ states.
The state complexity of both the counting and leader election protocols can be reduced to
$O(n^{\stateCountExponentMin})$
and
$O(n^{\LCsubprotocolStateCountConstantMin})$
respectively,
while increasing the time to $O(\log^2 n)$.
\opt{sub}{\\}

\opt{sub}{
    \noindent
    \textbf{Best student paper award eligibility:}
    Mahsa Eftekhari is a full-time student and made a significant contribution to the paper.

    \noindent
    \textbf{Full version of this paper:}
    \url{http://web.cs.ucdavis.edu/~doty/papers/escuppnlt.pdf}\\
    A time-stamped version is on arXiv,
    whose link (still being processed by arXiv at the time of DISC submission)
    can be found at \url{http://web.cs.ucdavis.edu/~doty/papers\#tech-reports}
}

\end{abstract}

\section{Introduction}
\label{sec:intro}

\emph{Population protocols}~\cite{AADFP06}
are networks that consist of computational entities called \emph{agents}
with no control over the schedule of interactions with other agents.
In a population of $n$ agents,
repeatedly a random pair of agents is chosen to interact,
each observing the state of the other agent before updating its own state.\opt{full}{\footnote{Using message-passing terminology, each agent sends its entire state of memory as the message.}}
They are an appropriate model for electronic computing scenarios such as sensor networks
and for ``fast-mixing'' physical systems such as
animal populations~\cite{Volterra26},
gene regulatory networks~\cite{bower2004computational},
and chemical reactions~\cite{SolCooWinBru08},
the latter increasingly regarded as an implementable ``programming language'' for molecular engineering,
due to recent experimental breakthroughs in DNA nanotechnology~\cite{chen2013programmable, srinivas2017enzyme}.

The \emph{(parallel) time} for some event to happen in a protocol is a random variable,
defined as the number of interactions, divided by $n$,
until the event happens.
This measure of time is based on the natural parallel model where
each agent participates in $\Theta(1)$ interactions each unit of time;
hence $\Theta(n)$ total interactions per unit time~\cite{AAE08}.
In this paper all references to ``time'' refer to parallel time.
The original model~\cite{AADFP06} stipulated that each agent is a finite-state machine,
with a state set $Q$ constant with respect to the population size,
and a transition function $\delta: Q \times Q \to Q \times Q$
indicating that if agents in states $a \in Q$ and $b \in Q$ interact,
and $\delta(a,b) = (c,d)$,
then they update respectively to states $c$ and $d$.
Many important distributed computing problems
provably require $\Omega(n)$ time for population protocols
to solve under the constraint of $O(1)$ states,
such as leader election~\cite{LeaderElectionDIST},
exact majority computation~\cite{alistarh2017time},
and many other predicates and functions~\cite{hcapflpp}.

This limitation on time-efficient computation with constant-memory protocols
motivates the study of population protocols with memory that can grow with $n$.
A recent blitz of impressive results has shown that
leader election~\cite{alistarh2017time, GS18, berenbrink2018simple, bilke2017brief}
and exact majority~\cite{AG15, alistarh2018space}
can be solved in $\mathrm{polylog}(n)$ time
using $\mathrm{polylog}(n)$ states.
Notably,
each of these protocols requires an approximate estimate of the population size $n$ to be encoded into each agent
(commonly a constant-factor upper bound on $\lfloor \log n \rfloor$ or $\lfloor \log \log n \rfloor$).

This \emph{partially} motivates our study of
the \emph{size counting problem}
of computing the population size $n$.
The problem is clearly solvable by an $O(n)$ time protocol
using a straightforward leader election:
Agents initially assume they are leaders and the count is 1.
When two leaders meet,
one agent sums their counts
while the other becomes a follower,
and followers propagate by epidemic the maximum count.
No faster protocol was previously known.

Our study is further motivated by the desire to understand the power of \emph{uniform} computation in population protocols.
All of the mentioned positive results
with $\omega(1)$ states~\cite{AG15, alistarh2017time, bilke2017brief, GS18, berenbrink2018simple, alistarh2018space, MocquardAABS2015, MocquardAS2016}
use a \emph{nonuniform} model:
given $n$,
the state set $Q_n$ and transition function $\delta_n: Q_n \times Q_n \to Q_n \times Q_n$
are allowed to depend arbitrarily on $n$,
other than the constraint that
$|Q_n| \leq f(n)$ for some ``small'' function $f$ growing as $\mathrm{polylog}(n)$.\footnote{
    Another constraint,
    sometimes explicitly stated, e.g., \cite{alistarh2017time},
    but usually implicit from the constructions,
    requires that if $|Q_n| = |Q_m|$ for $n \neq m$,
    then $Q_n = Q_m$ and $\delta_n = \delta_m$.
}
This nonuniformity is used to encode a value such as $\floor{\log n}$ into the cited protocols.\footnote{
    In these papers~\cite{AG15, alistarh2017time, bilke2017brief, GS18, berenbrink2018simple, alistarh2018space, MocquardAABS2015, MocquardAS2016},
    the role of the value $\log n$ (or $\log \log n$) is as a threshold to compare to some other integer $k$,
    which starts at 0 and increments,
    stopping some stage of the protocol when $k \geq \log n$.
    A na\"{i}ve attempt to achieve uniformity is to initialize the comparison threshold to some constant $c < \log n$,
    which is then updated by the agent with each interaction in such a way that $c$ ``quickly'' reaches some value $\geq \log n$.
    The challenge, however, is that prior to the event $c \geq \log n$,
    the comparison ``$k \geq c$?'' should never evaluate to true and cause an erroneous early termination of the stage,
    nor should a fast-growing $c$ ``overshoot'' $\log n$ and excessively increase the memory requirement.
}

We define a stricter \emph{uniform} variant of the model:
the same transition algorithm is used for all populations,
though the number of states may vary with the population size
(formalized with Turing machines; see Section~\ref{sec:model}.)
A uniform protocol can be deployed into \emph{any} population without knowing in advance the size,
or even a rough estimate of the size.
The original, $O(1)$-state model~\cite{AADFP06, AAE06, AAE08},
is uniform since there is a single transition function.
Because we allow memory to grow with $n$,
our model's power exceeds that of the original,
but is strictly less than that of the nonuniform model of most papers using $\omega(1)$ states.


\subsection{Contribution}
Our main result is a uniform protocol that,
with probability $\geq 1-\exactCountingConvergenceTimeProbErr$,
counts the population size,
converging in $\exactCountingConvergenceTimeConstant \log n \log \log n$ time
using $\stateCountConstant n^{\stateCountExponent}$ states
($\stateCountConstantLog + \stateCountExponent \log n$ bits),
without an initial leader
(all agents have an initially identical state).
The protocol is \emph{stabilizing}:
the output is correct with probability 1,
converging in expected time $\expectedTimeConstant \log n \log \log n$.\footnote{The time to reach a stable configuration, from which the output \emph{cannot} change, is $\Omega(n)$ for our protocol. We leave open the question of a protocol that reaches a stable configuration in sublinear time.}

A key subprotocol performs leader election
in time $O(\log n \log \log n)$ with high probability and in expectation.
It uses $O(n^{\LCsubprotocolStateCountConstant})$ states,
much more than the $O(\log \log n)$ known to be necessary~\cite{alistarh2017time} and sufficient~\cite{GS18}
for sublinear time leader election.
However,
it is uniform,
unlike all known sublinear-time leader election
protocols~\cite{alistarh2017time, GS18, berenbrink2018simple, bilke2017brief}.
It repeatedly increases the length of a binary string each agent stores,
where the protocol,
with probability at least $1 - O(1/n)$,
takes $O(\log n)$ time once the length of this string reaches $\approx \log n$.
The length increases in stages that each take $O(\log n)$ time.
Our main protocol doubles the string length each stage,
so takes $\log \log n$ stages
(hence $O(\log n \log \log n)$ time)
to reach length $\log n$.

The protocol generalizes straightforwardly to trade off time and memory:
by adjusting the rate at which the string length grows,
the convergence time $t(n)$ is $O(f(n) \log n)$,
where $f(n)$ is the number of stages required for the string length to reach $\log n$.
For example,
if the code length increments by 1 each stage,
then $f(n) = \log n$,
so $t(n) = \log^2 n$.
In this case the state complexity would be $O(n^{\stateCountExponentMin})$
for the full protocol
and $O(n^{\LCsubprotocolStateCountConstantMin})$
for just the leader election portion.
\opt{full}{(See Section~\ref{subsec:minimize-states}.)}
\opt{sub}{(See the full version of the paper for details.)}
By squaring the string length each stage,
$t(n) = \log n \log \log \log n$.
By exponentiating the string length,
$t(n) = \log n \log^* n$.
Even slower-growing $f(n)$ such as inverse Ackermann are achievable.
However,
the faster the string length grows each stage,
the more it potentially overshoots $\log n$,
increasing the space requirements.
For example,
for $t(n) = \log n \log \log \log n$ by repeated squaring,
the worst-case string length is $\log^2 n$,
meaning $2^{O(\log^2 n)} = n^{O(\log n)}$ states.
Multiplying length by a constant gives the fastest increase that maintains a polynomial number of states.

The number of states our protocol uses is very large compared to most population protocol results,
which typically have $\mathrm{polylog}(n)$ states.
However, it is worth noting that a different goalpost is germane for the size counting problem:
at least $n$ states are required,
since it takes $\log n$ bits merely to write the number $n$.
Our protocol uses a constant factor more bits: about $\stateCountExponent \log n$.
Chemical reaction networks are frequently cited as a real system for which population protocols are an appropriate model.
It is reasonable to object that
since each state corresponds to a different chemical species,
such a large number of states is unrealistic.
However, biochemistry provides numerous examples of heteropolymers,
such as nucleic acids and peptide chains
(linear polymers of amino acids that fold into proteins),
in which $c = O(1)$ basic monomer types (e.g., the 4 DNA bases A, C, G, T)
suffice to construct $c^k$ different polymer types consisting of $k$ monomers.
On the engineering side,
DNA strand displacement systems~\cite{SolSeeWin10}
can in principle construct and modify such information-rich ``combinatorial'' polymers
in a controllable algorithmic fashion,
for example simulating a Turing machine whose length-$k$ tape
is represented by $O(k)$ DNA strands~\cite{polymerStackDNA16}
or
searching for solutions to a quantified Boolean formula~\cite{thachuk2012space}.
The synthesis cost for such systems would scale with the number of \emph{bits} of memory
($O(1)$ molecules per bit stored),
thus only logarithmically with the total number of states.
It is thus reasonable to conjecture that reliable algorithmic molecular systems,
with moderately sized memories in each molecule,
are on the horizon.


\subsection{Related Work}
\label{subsec:related}

For the exact population size counting problem, the most heavily studied case is that of \emph{self-stabilization},
which makes the strong adversarial assumption that arbitrary corruption of memory is possible in any agent at any time,
and promises only that eventually it will stop.
Thus, the protocol must be designed to work from any possible configuration of the memory of each agent.
It can be shown that counting is \emph{impossible} without having one agent (the ``base station'')
that is protected from corruption~\cite{beauquier2007self}.
In this scenario
$\Theta(n \log n)$ time is sufficient~\cite{beauquier2015space} and necessary~\cite{AspnesBBS2016} for self-stabilizing counting.
Counting has also been studied in the related context of worst-case dynamic networks~\cite{IK14, KLO10, MCS13, LBBC14, CFQS12}.

\todo{DD: the paper \cite{CMNPS11} talks about population size, but I don't understand the exact problem studied and how it relates to ours. Othon and Paul, do you want to add a discussion of it?}
In the less restrictive setting in which all nodes start from a pre-determined state,
Michail \cite{M15} proposed a terminating protocol in which a pre-elected leader equipped with two $n$-counters computes an approximate count between $n/2$ and $n$ in $O(n\log{n})$ parallel time with high probability.
Regarding \emph{approximation} rather than exact counting,
Alistarh, Aspnes, Eisenstat, Gelashvili, Rivest~\cite{alistarh2017time} have shown a
uniform protocol that in $O(\log n)$ expected time converges to an approximation $n'$ of the true size $n$
such that with high probability $1/2 \log n \leq \log n' \leq 9 \log n$,
i.e., $\sqrt{n} \leq n' \leq n^9$.\footnote{
    We also require an approximate estimate of $n$ in the subprotocol that computes $n$ exactly, but it is not straightforward to adapt the technique of~\cite{alistarh2017time} to our setting. The state complexity would be higher, since our method of estimating $n$ obtains $n'$ such that $n \leq n' \leq n^6$. By squaring $n'$ obtained from~\cite{alistarh2017time} to ensure it is at least $n$, the result could be as large as $n^{18}$.
}

Key to our technique is a protocol, due to Mocquard, Anceaume, Aspnes, Busnel, and Sericola~\cite{MocquardAABS2015}.
Despite the title of that paper (``Counting with Population Protocols''),
it actually solves a different problem,
a generalization of the majority problem:
count the exact difference between ``blue'' and ``red'' agents in the initial population.
The protocol assumes an initial leader and
that each agent initially stores $n$ exactly.
In a follow-up work~\cite{MocquardAS2016},
Mocquard et al.~showed a \emph{uniform} protocol that,
for any $\epsilon > 0$,
computes an approximation of the relative proportion
(but not exact number)
of ``blue'' nodes in the population,
within multiplicative factor $(1 + \epsilon)$ of the true proportion.
The approximation precision $\epsilon$ depends on a constant number $m$,
which is encoded in the initial state.
They also describe a protocol to find the number of ``blue'' nodes in the population,
However, like~\cite{MocquardAABS2015},
this latter protocol is not uniform since the transition function encodes the exact value of $n$.

In a different network model,
Jelasity and Montresor~\cite{JM2004}
use a similar technique to ours that involves a fast ``averaging'' similar to~\cite{MocquardAABS2015,MocquardAS2016}.
However, they do arbitrary-precision rational number averaging,
so have a larger memory requirement (not analyzed).
They also assume each agent initially has a unique IDs.
Goldwasser, Ostrovsky, Scafuro, Sealfon~\cite{GOSS18} study a related problem in
a synchronous variant of population protocols:
assuming that both an adversary and the agents themselves have the ability to create and destroy agents
(similar to the more general model of chemical reaction networks),
using $\mathrm{polylog}(n)$ states,
they \emph{maintain} the population size within a multiplicative constant of a target size.
This is likely relevant to the exact and approximate size counting problems,
since the protocol of~\cite{GOSS18} must ``sense'' when the population size is too large or small and react.

\section{The model}
\label{sec:model}

\newcommand{\blank}{\llcorner\negthinspace\lrcorner}

The system consists of a population of $n$ distributed and anonymous
(no unique IDs)
\emph{agents}, also called \emph{nodes} or \emph{processes}, that can perform local computations.
Each agent is a multitape ($r$-tape) Turing Machine which is defined by a $6$-tuple $M = \langle Q, \Gamma, q_0, \epsilon, F, \delta \rangle$.
$Q$ is a finite set of \emph{TM states},
\todo{DD: I think we should allow blank in the tape alphabet. Or somehow indicate that the ability to tell where the end of the tape is has been ``built into the TM hardware''. Otherwise we have to talk about encodings of delimiters and it would get really ugly.}
$\Gamma$ is the binary \emph{tape alphabet} $\{0, 1\}$,
$q_0 \in Q$ is the \emph{initial TM state},
$F \subseteq Q$  is the set of \emph{final TM states},
and $\delta: Q \times \Gamma^r \rightarrow Q \times (\Gamma \times \{L,R,S\})^r$ is the \emph{TM transition function}, where $L$ is left shift, $R$ is right shift and $S$ is no shift.
We assume that $r$ is a fixed constant, i.e. independent of the population size.

We define three types of tapes. \textit{Input}, \textit{Output} and \textit{Work} tapes. The Input and Output tapes provide information from and to the other agent during an interaction. The Work tapes are used for storing data and for internal operations, which can be assumed to be additions, subtractions and multiplications (divisions can be performed via the Euclidean Division Algorithm, which divides two integers using additions and subtractions).

Let $r_i < r$ be the number of input tapes, $r_o < r$ the number of output tapes and $r_w < r$ the number of work tapes, where $r_i + r_o + r_w = r$.
For any $t \geq 0$ let $I(t), O(t)$ and $W(t)$ be $|V| \times r_i$, $|V| \times r_o$ and $|V| \times r_w$ matrices respectively, such that $I_{v,j}(t)$, $O_{v,j}(t)$ and $W_{v,j}(t)$ are the values of the $j$-th Input, Output and Work tapes respectively of the agent $v \in V$ at time $t$. Furthermore, for every $t \geq 0$, let $q(t)$ be a $|V|$-dimensional vector such that $q_{v}(t)$ is the \emph{state} (or \emph{agent-TM-configuration}) of $v \in V$ at time $t$.
We refer to $q(t)$ as the \emph{configuration} (or \emph{global-configuration}) at time $t$.
We say that a population protocol is \emph{leaderless} if $q_{v}(0) = q_0 \; \forall \; v \in V$, i.e.
all agents have the same state in the initial configuration.
\todo{MT: Should we also have to say that not only they have the same initial state, but also the tapes have the same values?
DD: state has now been defined to mean ``TM configuration'', not ``TM state'', so ``same state'' implies ``same tape values''.}
We also say that $I(0)$ is the population input at time 0.

Let $S$ be the finite set of binary strings $\{0,1\}^*$.
This model is defined on a population $V$ of agents and consists of an
\emph{input initialization function}
$\iota : S \rightarrow S^r \times Q$ and an
\emph{output function}
$\gamma : S^r \times Q \rightarrow D$
($D$ is the set of output values).
Initially, the values of the tapes of each agent are determined by the input initialization function $\iota$, and in every step $t+1 \geq 1$, a pair of agents interacts.
During an interaction $(a, b)$ between two agents at time $t+1$, each agent updates its state and copies the contents of its Output tapes to the Input tapes of the other agent $(O_{a,:}(t) \rightarrow I_{b,:}(t+1)$ and $O_{b,:}(t) \rightarrow I_{a,:}(t+1))$.
In addition, they update their states according to the (global) joint transition function $f: Q \times Q \ra Q \times Q$ as in standard population protocols.

We furthermore assume that each agent has access to independent uniformly random bits,
assumed to be pre-written on a special read-only tape
(so that we can use a deterministic TM transition function).
This is different from the traditional definition of population protocols,
which assumes a deterministic transition function.
Several papers~\cite{alistarh2017time, berenbrink2018simple}
indicate how to use the randomness built into the interaction scheduler to provide nearly uniform random bits to the agents,
using various \emph{synthetic coin} techniques,
showing that the deterministic model can effectively simulate the randomized model.
In the interest of brevity and simplicity of presentation,
we will simply assume in the model that each agent has access to a source of uniformly random bits.

\todoi{DD: We had previously said something like this: One way to remove this assumption is to use the fact that each agent in an interaction has probability exactly $\frac{1}{2}$ to be the sender. The bits from this tape can be used as random bits from the algorithm, other than the current, when needed, with only a slight delay due to the fact that an agent may require more bits than the total number of interactions it has experienced so far.

But this actually doesn't work. If agents store 1 as receiver and 0 as sender, then knowing the bits gives some information about the current state.
For example, if an agent is all 1's,
then since only receivers update the state in a one-way protocol like ours,
this agent probably is at a deeper level than others,
having gone through more state-changing interactions than other agents.}

\noindent \textbf{Memory requirements for \exactCounting\ protocol}.
In our main protocol, \exactCounting, the agents need $r_i = r_o = 3$ Input and Output tapes for storing the variables $\C$, $\LC$ and $\ave$. The memory requirements (number of bits) are
$|\C| = \Cconstant \log{n}$,
$|\LC| = \LCconstant \log{n}$ and
$|\ave| = \aveconstant \log{n}$.
In addition, three Work tapes are needed in order to store the variables $\M$, $\thecount$  and the constants $\isLeader$ and $\phase$ (the first cell of a tape can be used for storing the boolean variable $\isLeader$, while $\phase$ can be stored after that cell). The memory requirements (number of bits) are $|\M| = 1 + \Mconstant \log{n}$, $|\thecount| = \countconstant \log{n}$ and $|\isLeader| + |\phase| = 1 + \log{\maxPhaseValue}$. Finally, two more Work tapes are needed in order to perform divisions between integers, using the Euclidean Division Algorithm. 

\noindent
\textbf{Terminology conventions.}
Throughout this paper, $n$ denotes the number of agents in the population.
The \emph{time} until some event is measured as the number of interactions until the event occurs, divided by $n$,
also known as parallel time.
This represents a natural model of time complexity in which we expect each agent to have $O(1)$ interactions per unit of time, hence across the whole population, $\Theta(n)$ total interactions occur per unit time.
All references to ``time'' in this paper refer to parallel time.
$\log n$ is the base-2 logarithm of $n$,
and $\ln n$ is the base-$e$ logarithm of $n$.

For ease of understanding,
we will use standard population protocol terminology and not refer explicitly to details of the Turing machine definition except where needed.
Therefore a \emph{state} always refers to the TM initial configuration of an agent
(leaving out TM state and tape head positions since these are identical in all initial configurations),
a \emph{configuration} $\vec{c}$ refers to the length-$n$ vector giving the state of each agent,
and \emph{transition function} refers to the function computing the next state of an agent,
taking as input its state and the other agent's state,
by running its Turing machine until it halts.
An \emph{epidemic}~\cite{AAE08} is a subprotocol of the form
$\delta(i,u) = (i,i)$ starting with one agent with $i$ (``infected'')
and all other $n-1$ agents with $u$ (``uninfected''),
which in $O(\log n)$ expected time converts all agents to $i$.

\todoi{make sure to mention that we subsequently use standard population protocol terminology: state means state of agent (initial configuration of Turing machine), configuration means length-$n$ vector describing state of each agent, agents have ``fields'', which is shorthand for worktape where strings/integers/other discrete data is stored}


\subsection{Stabilization and convergence}

A protocol \emph{converges} when it reaches a configuration where all agents have the same output,
which does not subsequently change.
In our main protocol,
agents have a field \thecount,
and convergence to the correct output
occurs when each agent has written the value $n$ into \thecount\ for the last time.
Configuration $\vec{c}$ is \emph{stable} if every agent agrees on the output,
and no configuration reachable from $\vec{c}$ has a different output in any agent.
A protocol \emph{stabilizes} if,
with probability 1,
it eventually reaches a correct stable configuration.
Using this terminology,
a protocol \emph{stably solves the exact size counting problem}
if,
for all $n \in \mathbb{Z}^+$,
with probability 1,
on a population of $n$ agents,
the protocol converges to output $n$ and enters a stable configuration.

If the number of configurations reachable from the initial configuration is finite,
then stabilization is equivalent to requiring that,
for every configuration $\vec{c}$ reachable from the initial configuration,
a correct stable configuration is reachable from $\vec{c}$.
It is also equivalent to saying that every fair execution reaches a correct stable configuration,
where an execution is \emph{fair} if every configuration that is infinitely often reachable in it is infinitely often reached.
Although our protocol as defined has an infinite number of reachable configurations,
this is done solely to make the analysis simpler,
and it can easily be modified to be finite
(see explanation of the \uniqueID\ protocol in Section~\ref{subsec:uniqueID}).

\newcommand{\polylog}{\mathrm{polylog}}
\newcommand{\N}{\mathbb{N}}

\section{Exact Population Size Counting} \label{sec:protocol}

This section is devoted to proving the main theorem of our paper:

\begin{theorem}\label{thm:main}
    There is a leaderless, uniform population protocol
    that stably solves the exact size counting problem.
    With probability at least $1 - \exactCountingConvergenceTimeProbErr$,
    the convergence time is at most
    $\exactCountingConvergenceTimeConstant \ln n \log \log n$,
    and each agent is uses
    $\stateCountConstantLog + \stateCountExponent \log n$
    bits of memory.
    The expected time to convergence is at most
    $\expectedTimeConstant \ln n \log \log n$.
\end{theorem}

The stabilization time can be much larger, up to $O(n)$.
(See Section~\ref{subsec:expected-time}.)
Theorem~\ref{thm:main}
follows from Theorems~\ref{thm:main:whp} and~\ref{thm:main:exp},
which respectively cover the ``with high probability'' and ``stabilization and expected time''
parts of Theorem~\ref{thm:main}.

\begin{algorithm}[ht]
	\floatname{algorithm}{Protocol}
	\caption{$\exactCounting(\rec, \sen)$}
	\label{protocol:counting}
	\begin{algorithmic}[100]		
        \State{}
		\LeftComment {state: strings $\C$ (code), $\LC$ (leader code), Bool $\isLeader$, ints $\M$, $\ave$, $\thecount$, $\phase$}
		\LeftComment {initial state of agent: $\C$ = $\LC$ = $\varepsilon$, $\isLeader = \true$, $\M = \ave = \thecount = \phase = 1$}
    	\State{$\uniqueID(\rec, \sen)$}
		\State{$\electLeader(\rec, \sen)$}
        \If {$\rec.\LC = \sen.\LC$} \Comment{separate restarts under different leaders}
    		\State{$\averaging(\rec, \sen)$}
            \State{$\timer(\rec, \sen)$}
        \EndIf
	\end{algorithmic}
\end{algorithm}

The protocol is \exactCounting.
There are four main subprotocols:
\uniqueID,
\electLeader,
\averaging,
and
\timer,
each discussed in detail in later subsections.
\exactCounting\ runs in parallel on all agents,
but within an agent,
each subprotocol runs sequentially
(for correctness each subprotocol must run in the given order).
Most state updates use one-way rules for selected agents sen (\emph{sender}) and rec (\emph{receiver}).
The only rule that is not one-way is \averaging, in which both sender and receiver update their state.
In all other cases, only the receiver potentially updates the state.

\noindent
\textbf{High-level overview of \exactCounting\ protocol.}
\uniqueID\ eventually assigns to every agent a unique id,
represented as a binary string called a \emph{code} $\C$.
\uniqueID\ requires $\Omega(n)$ time to converge,
but it does not need to converge before it can be used by the other subprotocols.
In fact, in other subprotocols, agents do not use each others' codes directly.
Agents also have a longer code called a \emph{leader code} $\LC$,
such that $2|\C| = |\LC|$ and,
for any candidate leader,
$\C$ is a prefix of $\LC$.
\electLeader\ elects a leader by selecting the agent whose leader code is lexicographically largest.
The code length $|\C|$ will eventually be at least length $\log n$,
so $|\C|$ can be used to
estimate an upper bound $\M$ on the value $\lowerBoundOnM$ to within a polynomial factor.
\averaging\ uses $\M$ in a leader-driven protocol that counts the population size exactly,
which is correct so long as $\M \geq \lowerBoundOnM$, by using a fast averaging protocol similar to the one studied by Mocquard et al.~\cite{MocquardAABS2015}.
\averaging\ must be restarted by the upstream \uniqueID\ subprotocol many times,
and in fact will be restarted beyond the $O(\log n \log \log n)$ time bound we seek.
However, within $O(\log n \log \log n)$ time, \averaging\ will converge to the correct population size.
Subsequent restarts of \averaging\ will re-converge to the correct output,
but prior to convergence will have an incorrect output.
\timer\ is used to detect when \averaging\ has likely converged,
waiting to write output into the \thecount\ field of the agent.
This ensures that after the correct value is written,
on subsequent restarts of \averaging,
the incorrect values that exist before \averaging\ re-converges
will not overwrite the correct value recorded during the earlier restart.

\subsection{\uniqueID}
\label{subsec:uniqueID}

We assume that two subroutines are available:
For $x,y \in \{0,1\}^*$,
$\append(x,y)$ returns $xy$,
and for $m \in \mathbb{Z^+}$,
$\randbits(m)$ returns a random string in $\{0,1\}^m$.

\begin{algorithm}[ht]
	\floatname{algorithm}{Subprotocol}
	\caption{$\uniqueID(\rec, \sen)$}
	\begin{algorithmic}[100]
        \If {$|\rec.\C| < |\sen.\C|$} \Comment{If receiver's code shorter than sender's, make same length.}
		    \State {$\extendCode(\rec, |\sen.\C| - |\rec.\C|))$}
		\EndIf

        \If {$\rec.\C = \sen.\C$}  \Comment{If codes are the same, double the length.}
		    \State {$\extendCode(\rec, \max(1, |\rec.C|))$}
		\EndIf
	\end{algorithmic}
\end{algorithm}

\begin{algorithm}[ht]
	\floatname{algorithm}{Subroutine}
	\caption{$\extendCode(\rec, \textrm{numBits})$}
	\begin{algorithmic}[100]
        \If {$\rec.\isLeader$} \Comment{extend LC by twice numBits; take new C bits from LC}
    		\State {$\textrm{newLC} \gets \append(\rec.\LC, \randbits(2 \cdot \textrm{numBits}))$}
            \State {$\setNewLeaderCode(\rec, \textrm{newLC})$} \Comment{described in Subsection~\ref{subsec:leaderElection}}
    		\State {$\rec.\C \gets \append(\rec.\C, \textrm{newLC}[(|\rec.\C|+1)\ ..\ (|\rec.\C|+\textrm{numBits})])$}
        \Else
            \State {$\rec.\C \gets \append(\rec.\C, \randbits(\textrm{numBits}))$}
		\EndIf
	\end{algorithmic}
\end{algorithm}

\begin{algorithm}[ht]	
	\floatname{algorithm}{Subroutine}
	\caption{$\setNewLeaderCode(\rec, \textrm{newLC})$}
	\label{protocol:election}
	\begin{algorithmic}[100]
        \State {$\rec.\LC \gets \textrm{newLC}$}
        \LeftComment{restart Timer and Average protocols whenever $\LC$ changes.}
		\State {$\rec.\phase \gets 1$}
        \State {$\rec.\M \gets 3 \cdot 2^{3 |\rec.\C|}$ }
		\If {$\rec.\isLeader$}
		    \State {$\rec.\ave \gets \rec.\M$}
        \Else
            \State {$\rec.\ave \gets 0$}
		\EndIf
	\end{algorithmic}
\end{algorithm}


\uniqueID\ can be viewed as traversing a labeled binary tree,
until all agents reach a node unoccupied by any other agent.
We say the \emph{level} is the maximum depth (longest code length) of any agent in the population.
Initiating a new level happens when two agents with the same code interact.
The receiver doubles the length of its code with uniformly random bits,
going twice as deep in the tree.
To ensure each agent reaches the new level quickly,
agents at deeper levels recruit other agents to that level by epidemic,
which generate random bits to reach the same code length.

The key property of this protocol is that,
in any level $\ell < \log n$,
only $O(\log n)$ time is required to increase the level.
Since we double the level when it increases,
$\log \log n$ such doublings are required to reach level $\geq \log n$,
so $O(\log n \log \log n)$ time.
Lemma~\ref{lem:quick} formalizes this claim,
explaining how the length-increasing schedule can be adjusted
to achieve a trade-off between time and memory.

\noindent
{\bf Number of reachable configurations.}
Since all codes are generated randomly,
the number of reachable configurations is infinite.
This choice is merely to simplify analysis,
allowing us to assume that all agents at a level have uniformly random codes.
However, if a finite number of reachable configurations is desired
(so that, for instance, the definition of stabilization we use is equivalent to definitions based on reachability),
it is possible to modify \uniqueID\ so that
when two agents with the same code meet,
they \emph{both} append bits that are guaranteed to be different.
The protocol still works in this case
and in fact takes strictly less expected time for the codes to become unique.
Viewing two agents with compatible codes
(i.e., one code is a prefix of the other)
as equivalent,
each new level increases the number of equivalence classes by 1.
Thus it is guaranteed that all agents will converge on unique codes of length at most $n-1$,
implying the reachable configuration space is finite.





\begin{figure}[th]
	\centering
	\begin{subfigure}[b]{\textwidth}
        \includegraphics[width=\textwidth, draft=false]{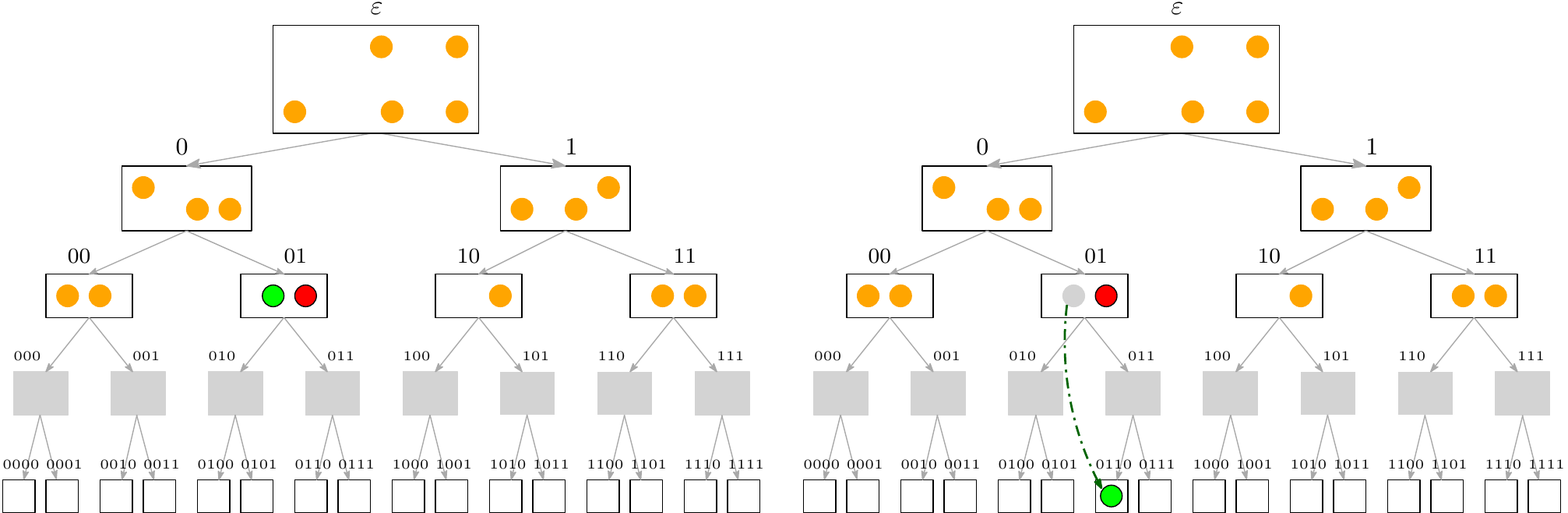}
        \caption{When two agents in the same node at level $i$ interact, receiver moves to a random descendant at level $2i$.}
        \label{fig:tree-rule1}
    \end{subfigure}
    \begin{subfigure}[b]{\textwidth}
        \includegraphics[width=\textwidth, draft=false]{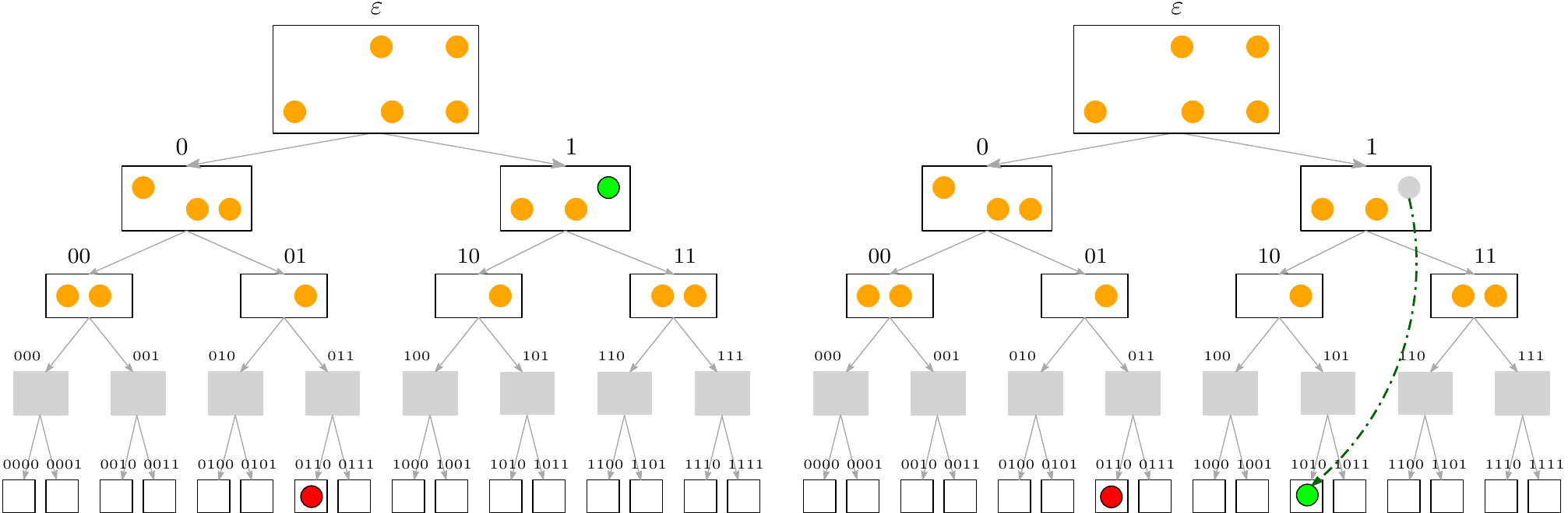}
        \caption{When sender is in a deeper level of the tree, receiver moves to a random descendant in its own subtree at the sender's level.}
        \label{fig:tree-rule2}
    \end{subfigure}
	\caption{Agents moving through the binary tree (i.e., choosing binary codes) in accordance with the \uniqueID\ subprotocol.}
	\label{fig:tree}
\end{figure}

\opt{full}{
    The following lemma is essentially Lemmas 1 and 2 from the paper \cite{AAE08}.
    However, that paper does not state how the various constants are related, which we require for our proofs.
    We recapitulate their proof, deriving those relationships explicitly.

    \begin{lemma}[\cite{AAE08}]\label{lem:epidemic}
    Let $T$ denote the time to complete an epidemic.
    Then
    $\E[T] \leq \epidemicExpectedTimeConstant \ln n$,
    $\Pr[T < \frac{1}{4} \ln n] < 2 e^{-\sqrt{n}}$,
    and for any $\alpha_u > 0$,
    $\Pr[T > \alpha_u \ln n] < 4 n^{- \alpha_u/4+1}$.
    \end{lemma}

    \begin{proof}
        We begin by showing $\Pr[T > \alpha_u \ln n]< 4 n^{- \alpha_u/4+1}$.
        Suppose we have $\alpha_u n \ln n$ interactions,
        starting with one infected agent.
        We want to bound the probability that any agent remains uninfected.
        The second half of an epidemic (after exactly $n/2$ agents are infected)
        has equivalent distribution to the first,
        so we analyze just the second half, bounding the probability it requires more than $(\alpha_u/2) n \ln n$ interactions.
        When half of the agents are infected,
        each interaction picks an infected sender with probability at least $1/2$.
        The number of interactions to complete the epidemic is then stochastically dominated by a binomial random variable $\mathcal{B}((\alpha_u/2) n \ln n, 1/2)$, equal to the number of heads after $(\alpha_u/2) n \ln n$ coin flips if $\Pr[\text{heads}] = 1/2$.

        Let $\mu = E[\mathcal{B}((\alpha_u/2) n \ln n, 1/2)] = (\alpha_u/4) n \ln n$ and $\delta = 2/\sqrt{n}$.
        By the Chernoff bound~\cite[Corollary 4.10]{mitzenmacher2005probability},
        \begin{eqnarray*}
    		\Pr\left[ \mathcal{B}((\alpha_u/2) n \ln n, 1/2) < (1-\delta) \mu \right]
    		& < &
            e^{- \delta^2 \mu}
            = 
            e^{- (4/n) (\alpha_u/4) n \ln n}
            = 
            e^{- \alpha_u \ln n}
            =
            n^{- \alpha_u}.
    	\end{eqnarray*}
    	So with probability at least $1-n^{-\alpha_u}$,
    	more than $((1-\delta) \alpha_u/4) n \ln n > (\alpha_u/8) n \ln n$ (since $\delta < 1/2$) interactions involve an infected sender.
    	
    	To complete the proof of the time upper bound,
    	we need to bound the probability that these $(\alpha_u/8) n \ln n$ interactions fail to infect all agents.
    	Conditioned on each interaction having an infected sender,
    	the random variable giving the number of interactions until all agents are infected is equivalent to the number of collections required to collect the last $n/2$ coupons out of $n$ total.
    	Angluin et al.~\cite[Lemma 1]{AAE08} showed that for any $\beta$,
    	it takes more than $\beta (n/2) \ln (n/2) < (\beta/2) n \ln n$ collections to collect $n$ coupons with probability at most $n^{-\beta+1}$.
    	Let $\beta = \alpha_u / 4$.
    	Then $\Pr[(\alpha_u/8) n \ln n \text{ interactions fail to infect every agent}] \leq n^{-\beta+1} < n^{-\alpha_u/4+1}$.
        By the union bound on the events
        \emph{``fewer than $(\alpha_u/8) n \ln n$ interactions involved an infected sender''}
        and
        \emph{``$(\alpha_u/8) n \ln n$ interactions fail to infect every agent''},
        the second half of the epidemic fails to complete within $(\alpha_u/2) n \ln n$ interactions
        with probability at most $n^{-\alpha_u} + n^{-\alpha_u/4+1} < 2 n^{-\alpha_u/4+1}$.

        Again by the union bound
        on the events
        \emph{``first half of the epidemic takes more than $(\alpha_u/2) n \ln n$ interactions''}
        and
        \emph{``second half of the epidemic takes more than $(\alpha_u/2) n \ln n$ interactions''},
        the whole epidemic takes more than $\alpha_u n \ln n$ interactions with probability at most $4 n^{-\alpha_u/4+1}$.

        To show $\Pr[T < \frac{1}{4} \ln n] < 2 e^{-\sqrt{n}}$,
        we note that Lemma 1 of~\cite{AAE08}
        shows that if $S_n$ is the number of times a coupon must be collected to collect all coupons,
        then $\Pr[S_n < \frac{1}{4} n \ln n] < 2 e^{-\sqrt{n}}$.
        The proof says $2 e^{-\Theta(\sqrt{n})}$,
        but inspection of the argument reveals that the big-$\Theta$ constant can be assumed to be 1.
        In this case,
        applying the coupon collector argument to the epidemic,
        since we are proving a time lower bound,
        if we assume that \emph{every} interaction involves an infected sender,
        this process stochastically dominates the real epidemic.
        Thus $\Pr[T < \frac{1}{4} \ln n] < 2 e^{-\sqrt{n}}$.

        To analyze the expected time,
        observe that when $k$ agents are infected,
        the probability that the next interaction infects an uninfected agent is
        $\frac{k (n-k)}{n(n-1)} > \frac{k(n-k)}{n^2}$,
        so expected interactions until an infection at most
        $\frac{n^2}{k(n-k)}$.
        By linearity of expectation,
        the expected number of interactions to complete the epidemic is
        \begin{eqnarray*}
            \sum_{k=1}^{n-1} \frac{n^2}{k(n-k)}
            &=&
            2 n^2 \sum_{k=1}^{n/2} \frac{1}{k(n-k)} \ \ \ \ \ \ \text{sum is symmetric about middle index}
            \\&<&
            2 n^2 \sum_{k=1}^{n/2} \frac{1}{kn/2}
            = 
            \epidemicExpectedTimeConstant n \sum_{k=1}^{n/2} \frac{1}{k}
            < 
            \epidemicExpectedTimeConstant n \ln (n/2+1) < \epidemicExpectedTimeConstant n \ln n,
        \end{eqnarray*}
        i.e., expected time $< \epidemicExpectedTimeConstant \ln n$.
    \end{proof}
}

The next lemma bounds the time for \uniqueID\ to reach level at least $\log n$,
assuming a generalized way of increasing the level,
defining $f(n)$ to be the number of times the level must increase before reaching at least level $\log n$.
Afterwards we state a corollary for our protocol,
which doubles the level whenever it increases,
so $f(n) = \log \log n$.
By using this lemma with different choices of $f$,
one can obtain a tradeoff between time and space;
if the level increases more
(corresponding to a slower-growing $f$)
this takes less time to reach level at least $\log n$,
but may overshoot $\log n$ and use more space.

Intuitively,
the lemma is proven by observing that the worst case is that the current level is $\log(n) - 1$.
It takes  $O(\log n)$ time for all agents not yet at that level to reach it by epidemic.
At that point the worst case is that codes are distributed to maximize expected time:
exactly $n/2$ codes each shared by two agents.
Then the expected time is constant for the first interaction between two such agents,
starting the next level.
Thus it takes time $O(\log n)$ to increase the level,
hence $O(f(n) \log n)$ time for the level to increase from 0 to at least $\log n$.
\opt{sub}{
    A proof is contained in the full version of this paper.
}

\begin{lemma}\label{lem:quick}
    For all $n\in\N$, define $f(n)$
    to be the number of times \uniqueID\ must increase the level (last line of \uniqueID)
    to reach level at least $\log n$.
	For all $\alpha > 0$,
    in time $5 \alpha f(n) \ln n$,
    all agents reach level at least $\log n$
    with probability at least
    $1 - 5 f(n) n^{-\alpha}$.
\end{lemma}

\opt{full}{
    \begin{proof}
        Imagine an alternate process where at each level agents wait until all other agents also reach the same level before enabling transitions that start the next level (where two agents with the same code meet and the receiver will double its code length).
        The time for such a process stochastically dominates the time for our protocol, so we can use its time as an upper bound for our protocol.
        It suffices to show that,
        when all agents are at the same level,
        it takes constant time to start the next level.
        After initializing a level, by Lemma~\ref{lem:epidemic},
        the new code length will spread by epidemic
        in time $\alpha_u \ln{n}$
        with probability at least $1-4n^{-4\alpha_u+1}$.

        Assume all agents are currently at level $i$.
        Denote by $S_j$ the number of agents in node $j$ of the tree at level $i$
        (i.e., they have the $j$'th code in $\{0,1\}^i$ in lexicographic order).
        The probability that the next interaction is between two agents at the same node (having equal codes) is minimized when $S_j = S_{j'} = n/2^i$ for all $1 \leq j,j' \leq 2^i$.
        Then for all $0 \leq i < \log n$,
        if the current level is $i$,
    	\begin{eqnarray*}
    		\Pr\left[ \text{next interaction initializes new level} \right]
    		& = &
            \frac
    		{\sum_{j=1}^{2^i}\binom{S_j}{2}}
    		{\binom{n}{2}}
            = 
            \frac
    		{\sum_{j=1}^{2^i}\binom{n / 2^i}{2}}
    		{\binom{n}{2}}
            \\ & = &
            \frac
    		{2^i (n/2^i) (n/2^i-1)}
    		{n(n-1)}
    		= 
    		\frac
    		{n/2^i-1}
    		{n-1}
    		\\ &\geq&
    		\frac{n/2^{\log(n) - 1} - 1}{n-1}
    		=
    		\frac{1}{n-1}
    		>
    		\frac{1}{n}.
    	\end{eqnarray*}
    	Therefore, the expected number of interactions to start a new level is $\leq n$, equivalently parallel time 1.
        This is a geometric random variable with success probability at least $\frac{1}{n}$.
        Then at any level $i < \log n$,
        for any $\alpha_u' > 0$,
        \begin{eqnarray*}
    		\Pr\left[ \text{initializing next level take more than $\alpha_u' n \ln n$ interactions} \right]
    		& = &
            \left(1 - \frac{1}{n} \right)^{\alpha_u' n \ln n}
            \\ & < &
            e^{- \alpha_u' \ln n}
    		= 
    		n^{- \alpha_u'}.
    	\end{eqnarray*}


        By Lemma~\ref{lem:epidemic},
        for any $\alpha_u > 0$,
        more than $\alpha_u \ln{n}$ time is required for all agents to reach this level by epidemic
        with probability at most $4 n^{- \alpha_u/4+1}$.
        By the union bound over this event and the event
        \emph{``once all agents are at a level, it takes more than time $\alpha_u' \ln n$ to start a new level''}
        (shown above to happen with probability at most $n^{-\alpha_u'}$),
        the time spent at each level is more than
        $\alpha_u \ln n + \alpha_u' \ln n = (\alpha_u + \alpha_u') \ln n$
        with probability at most
        $4 n^{- \alpha_u/4+1} + n^{-\alpha_u'}$.
        Given $\alpha > 0$,
        let $\alpha_u' = \alpha$ and $\alpha_u = 4 (\alpha+1)$.
        Then this probability bound is $4 n^{- \alpha_u/4+1} + n^{-\alpha_u'} = 5 n^{-\alpha}$.

        By the union bound over all $f(n)$ levels visited in the tree,
        it takes more than time
        $f(n) (\alpha \ln n + 4 \alpha \ln n) = 5 \alpha f(n) \ln n$
        time to reach level at least $\log n$
        with probability at most
        $5 f(n) n^{-\alpha}$.
    \end{proof}
}

The next corollary is specific to our level-doubling schedule,
used throughout the rest of the paper,
corresponding to $f(n) = \log \log n$ in Lemma~\ref{lem:quick}.

\begin{corollary}\label{cor:quick}
	In the \uniqueID\ protocol,
    for all $\alpha > 0$,
    in $5 \alpha \ln n \log \log n$ time all agents reach level at least $\log n$
    with probability at least
    $1 - \frac{5 \log \log n}{n^{\alpha}}$.
\end{corollary}

The previous results show \uniqueID\ quickly gets to level $\log n$.
The next lemma states that it does not go too far past $\log n$.
Intuitively,
if the level is $2 \log n$,
there are $n^2$ possible codes
chosen uniformly at random among $n$ agents,
a standard birthday problem with probability $\frac{1}{e}$ of a collision,
which drops off polynomially with the level
beyond $2 \log n$.
\opt{sub}{A proof is in the full version of the paper.}

\begin{lemma}\label{lem:all_length}
    Let $\epsilon>0$.
    If the current level is $(2+\epsilon) \log{n}$,
    all codes are unique
    with probability at least $1 - \frac{1}{n^{\epsilon}}$.
\end{lemma}

\opt{full}{
    \begin{proof}
        Consider the agents in order for agent $1, 2, \ldots$.
        The code of agent $i$ collides with the code of some agent $1,2,\ldots,i-1$ with probability $\frac{i-1}{c}$,
        where $c = 2^{(2+\epsilon) \log n} = n^{2+\epsilon}$ is the number of available codes.
        Then by the union bound,
    	\begin{eqnarray*}
    		\Pr[\text{at least one collision}]
    		& \leq &
    		\sum_{i=1}^{n}\frac{i-1}{c}
            = 
    		\frac{n(n-1)}{2c}
            < 
            \frac{n^2}{n^{2+\epsilon}}
            = 
            \frac{1}{n^{\epsilon}}. \qedhere
    	\end{eqnarray*}
    \end{proof}
}

However, since the code length doubles when it changes,
not all values of $\epsilon$ in Lemma~\ref{lem:all_length} correspond to a level actually visited.
It could overshoot by factor two,
giving the following.

\begin{corollary}\label{cor:all_length}
Let $\epsilon>0$.
The eventual code length of each agent is $< (4+2\epsilon) \log{n}$
with probability at least $1 - \frac{1}{n^{\epsilon}}$.
\end{corollary}

\opt{full}{
    \begin{proof}
        Recall that a new level of the tree is initiated when two agents with the same code interact.
        Since the level is doubled in this case, the code lengths exceed $(4+2\epsilon) \log n$ if there was a duplicate code at the power-of-two level $k$ such that
        $(2+\epsilon) \log n \leq k < (4+2\epsilon) \log n \leq 2k$.
        Over all $k$ satisfying this inequality,
        the probability of a duplicate code is largest if $k = (2+\epsilon) \log n$.
        Applying Lemma~\ref{lem:all_length}
        gives the stated probability bound.
    \end{proof}
}

\subsection{\electLeader}
\label{subsec:leaderElection}

\begin{algorithm}[ht]	
	\floatname{algorithm}{Subprotocol}
	\caption{$\electLeader(\rec, \sen)$}
	\label{protocol:election-one-way}
	\begin{algorithmic}[100]
		\State{}
        \State {$p \gets \min(|\rec.\LC|, |\sen.\LC|)$}
		\If {$\rec.\LC[1..p]$ lexicographically precedes $\sen.\LC[1..p]$}
		    \LeftComment{Propagate by epidemic the lexicographically greatest leader code}
    		\State {$\rec.\isLeader \gets \false$}
    		\State {$\setNewLeaderCode(\rec, \sen.\LC)$}
		\EndIf
		\If {(not $\rec.\isLeader$) and ($|\rec.\LC| < |\sen.\LC|$)}
		    \LeftComment{Ensure all leader codes eventually have equal length}
		    \State {$\setNewLeaderCode(\rec, \sen.\LC)$}
		\EndIf
	\end{algorithmic}
\end{algorithm}

\electLeader\ works by propagating by epidemic the ``winning'' leader code,
where a candidate leader drops out if they see an agent (whether leader or follower)
with a leader code that beats its own.
The trick is to define ``win''.
We compare the shorter leader code with the same-length prefix of the other.
If they disagree,
the lexicographically largest wins.
To ensure all leader code lengths are eventually equal,
a follower with the shorter leader code replaces it with the longer one.\footnote{
Leaders with shorter codes do not replace with longer codes,
because it may be that after adding new random bits to get to the current population level,
that agent would have the lexicographically largest leader code.
This is because a leader with a shorter leader code $\LC$ also has a shorter code $\C$,
so will eventually catch up in leader code length through the \uniqueID\ protocol.
}

The next lemma shows that the leader is probably unique when the population reaches level at least $\log n$.
Let $k \in \N$ be such that
$\log n \leq k$.
When the candidate leaders generate new values of $\LC$ upon reaching level $k$,
$|\LC| = 2k \geq 2 \log n$.
Since there are at least $n^2$ strings of length $2k \geq 2 \log n$,
in the worst case,
even if all $n$ agents remain candidate leaders at that time,
the probability that the lexicographically greatest leader code is duplicated is at most $\frac{1}{n}$.\footnote{
    This almost looks like a birthday problem, but we don't need all the remaining candidate leaders to have a unique leader code, only that the \emph{largest} code appears only once.
}
Thus, with probability at least $1 - \frac{1}{n}$,
one unique leader has the maximum leader code,
and in $O(\log n)$ time this leader code reaches the remaining candidate leaders by epidemic,
who drop out.
\opt{sub}{A proof is given in the full version of this paper.}


\begin{lemma}\label{lem:uniqueleader}
	At any level $\geq \log n$, with probability $\geq 1 - \frac{1}{n}$, there is a unique leader.
\end{lemma}

\opt{full}{
    \begin{proof}
    	Every remaining candidate leader at level $\geq \log n$ has a leader code with at least $2\log n$ bits.
    	We say $i$ and $j$ \emph{collide} if agents $i$ and $j$ are candidate leaders with the same leader code.
    	Let $X_{i,j}$ be the indicator variable: 
    	
    	$$
    	X_{i,j} = \left\{
    	\begin{array}{ll}
    	1 & \text{if $i$ and $j$ collide} \\
    	0 & \text{otherwise}
    	\end{array}
    	\right.
    	$$
    	
    	Let $c \geq 2^{2 \log n} = n^2$ be the number of possible leader codes.
        Note $\Pr[X_{i,j} = 1] = \frac{1}{c}$.
    	Let $X_i = \sum_{j \neq i}{X_{i,j}}$;
        by linearity of expectation $\E[X_{i}] = \frac{n-1}{c}$.
    	Since $c > n^2$,
        $\E[X_i] < \frac{n-1}{n^2} < \frac{1}{n}$.
    	By Markov's inequality,
    	$
        \Pr\left[ X_i \geq 1 \right] \leq \frac{1}{n},
        $
        and the event $X_i \geq 1$ is equivalent to the event that the leader code of agent $i$ is not unique.
    	If we set $i$ to be the agent with the lexicographically greatest leader code of any remaining candidate leader,
        we conclude that leader is unique with probability $\geq 1 - \frac{1}{n}$.
    \end{proof}
}

By Corollary~\ref{cor:quick},
the protocol reaches level $\log n$ in $O(\log n \log \log n)$ time.
Thus,
by Lemma~\ref{lem:uniqueleader},
with high probability, in time $O(\log n \log \log n)$
\electLeader\ converges (the second-to-last candidate leader is eliminated).
Unfortunately \electLeader\ is not terminating:
the remaining leader does not know when it becomes unique.
Thus it is not straightforward to compose it with the downstream protocols
\averaging\ and \timer.
Standard techniques for making the protocol terminating with high probability,
such as setting a timer for a termination signal that probably goes off only
after $K \log n \log \log n$ time for a large constant $K$,
do not apply here, because when we start we don't know the value $\log n \log \log n$.
Thus, it is necessary,
each time a leader adds to its code length,
to restart the downstream protocols the existence of a unique leader.\footnote{
    One might imagine restarts could be tied to the elimination of candidate leaders,
    which stops within $O(\log n \log \log n)$ time,
    rather than the extending of codes,
    which persists for $\Omega(n)$ time.
    However, the leader may become unique \emph{before} level $\log n$,
    when $|\C| < \log n$,
    so $\M = 3 \cdot 2^{2|\C|} < 3 \cdot n^3$
    is not sufficiently large to ensure correctness and speed of \averaging.
    (See Lemma~\ref{lem:round}, which is applied with $c=1$.)}
This is done in \setNewLeaderCode,
which is actually called by both \electLeader\ \emph{and} \uniqueID,
since extending $\C$ for a leader also requires extending $\LC$,
to maintain that $2|\C| = |\LC|$.

\todoi{DD: There's a corollary in the comments below here that I commented out because we don't use it anywhere.}



\subsection{\averaging}
\label{subsec:averaging}

\begin{algorithm}[ht]
	\floatname{algorithm}{Subprotocol}
	\caption{$\averaging(\rec, \sen)$}
	\label{protocol:avg-one-way}
	\begin{algorithmic}[100]
		\State {$\rec.\ave, \sen.\ave \gets
		\ceil{  \dfrac{\rec.\ave + \sen.\ave}{2} } ,
		\floor{ \dfrac{\rec.\ave + \sen.\ave}{2} }$}
	\end{algorithmic}
\end{algorithm}

The previous subsections described how to set up a protocol (perhaps restarted many times)
to elect a leader and to produce a value $\M \geq \lowerBoundOnM$.
(With high probability we also have $\M \leq 3 \cdot n^{\Mconstant}$.)
Thus we assume the initial configuration of this protocol is one leader and $n-1$ followers,
each storing this value $\M$,
and that the goal is for all of them to converge to a value in $\ave$
such that $n = \floor{ \frac{\M}{\ave} + \frac{1}{2} }$.

There is an existing \emph{nonuniform} protocol~\cite{MocquardAABS2015}
that can do the following in $O(\log n)$ time.
Each agent starts with a bit $b \in \{0,1\}$ and a number $\M = \Omega(n^{3/2})$.
Let $n_b$ be the (unknown) number of agents storing bit $b$, so that $n_0+n_1=n$.
The agents converge to a state in which they all report the value $n_1 - n_0$,
the initial different in counts between the two bits.

Their protocol requires that $\M \geq \floor{n^{3/2} / \sqrt{2 \delta}}$
to obtain an error probability of $\leq \delta$.
The protocol is elegantly simple:
agents with $b=0$ start with an integer value $-\M$, while agents with $b=1$ start with an integer value $\M$,
and state space integers in the interval $\{-\M,-\M+1,\ldots,\M-1,\M\}$.
When two agents meet, they average their values, with one taking a floor and the other a ceiling in case the sum of the values is odd.
If an agent holds value $x$, that agent's output is reported as $\lfloor nx/\M + 1/2 \rfloor$,
i.e., $nx/\M$ rounded to the nearest integer.
This eventually converges to all agents sharing the population-wide average
$(n_1-n_0) \frac{\M}{n}$,
and the estimates of this average get close enough for the output to be correct within $O(\log n)$ time~\cite{MocquardAABS2015}.

Our protocol essentially inverts this,
starting with a known $n_0 = 1$ and $n_1 = n-1$,
computing the population size as a function of the average.
The leader starts with value $\ave = \M$,
and followers start with $\ave = 0$,
and the state space is $\{0,1,\dots,\M\}$.
The population-wide sum is always $\M$.\footnote{
    Think of the leader starting with $\M$ ``balls''.
    Interacting agents exchange balls until they have an equal number,
    or within 1.}
Eventually all agents have
$\ave = \lceil \frac{\M}{n} \rceil$ or $\lfloor \frac{\M}{n} \rfloor$,
which could take linear time in the worst case.
We show below that
with probability at least $1 - n^{-c}$,
in $O(\log n)$ time,
all agents' ave values are within $n^c$ of $\frac{\M}{n}$.
Each agent reports the population size as
$\floor{\frac{\M}{\ave} + \frac{1}{2}}$.
This is the exact population size $n$
as long as $\M \geq 3n^{c+2}$
and $\ave$ is within $n^c$ of $\frac{\M}{n}$,
as the following lemma shows.
\opt{sub}{A proof is given in the full version of this paper.}

\renewcommand{\ceil}[1]{\left\lceil #1 \right\rceil}
\renewcommand{\floor}[1]{\left\lfloor #1 \right\rfloor}

\begin{lemma} \label{lem:round}
    Let $c \geq 0$.
	If $M \geq 3n^{c+2}$, and $x \in \left[ \frac{M}{n} - n^{c}, \frac{M}{n} + n^{c} \right]$,
	then
	$\floor{\frac{M}{x} + \frac{1}{2}} = n$.
\end{lemma}

\opt{full}{
    \begin{proof}
        Since $\floor{\frac{M}{x} + \frac{1}{2}}$ is monotone in $x$,
        it suffices to show this holds for the two endpoints of the interval.
        For the case
        $x = \frac{M}{n} - n^{c}$,
        since $x < \frac{M}{n}$,
        we have
        $n < \frac{M}{x}$,
        and
        \begin{eqnarray*}
    		\frac{M}{x}
    		&=&
    		\frac{M}{\frac{M}{n} - n^{c}}
    		=
    		\frac{M}{\frac{M-n^{c+1}}{n}}
    		=
    		\frac{Mn}{M-n^{c+1}}
    		\\&\leq&
    		\frac{Mn}{M-M/(3n)} \ \ \ \ \text{since $M \geq 3n^{c+2}$}
    		\\&=&
    		\frac{n}{1-1/(3n)}
    		=
    		\frac{n}{(3n-1)/(3n)}
    		=
    		\frac{3n^2}{3n-1}
    		=
    	    n + \frac{1}{3(3n-1)} + \frac{1}{3}
    	    <
    	    n + \frac{1}{2}.
    	\end{eqnarray*}
    	So $n < \frac{M}{x} < n+\frac{1}{2}$,
    	so $\floor{\frac{M}{x} + \frac{1}{2}} = n$.
    	In the case $x = \frac{M}{n} + n^{c}$,
    	a similar argument shows that
    	$n - \frac{1}{2} < x < n$.
    \end{proof}
}

The above results show that the count computed by \averaging\ is correct if $\M$ is sufficiently large
and $\ave$ is within a certain range of the true population-wide average $\frac{\M}{n}$.
The next lemma,
adapted from~\cite[Corollary 8]{MocquardAABS2015},
shows that each agent's \ave\ estimate quickly gets within that range.
That corollary is stated in terms of a general upper bound $K$ on how far each agent's
\ave\ field starts from the true population-wide average.
In our case,
this is given by the leader,
which starts with $\ave = \M$,
while the true average is $\frac{\M}{n}$,
so we choose $K=\M > \M-\frac{\M}{n}$
in Corollary 8 of~\cite{MocquardAABS2015},
giving the following.

\begin{lemma}[\cite{MocquardAABS2015}]
    For all $\delta \in (0,1)$
    and all $t \geq \ln(4M^2)$,
    with probability at least $1-\delta$,
    after time $t$,
    each agent's \ave\ field is in the interval
    $\left[ \frac{M}{n} - \sqrt{\frac{n}{2\delta}}, \frac{M}{n} + \sqrt{\frac{n}{2\delta}} \right]$.
\end{lemma}

\begin{corollary}\label{cor:ave-close}
    Let $c>0$ and let $\delta = \frac{1}{2n^{2c-1}}$.
    For all $t \geq \ln(4M^2)$,
    with probability at least $1-\delta$,
    within time $t$,
    each agent's \ave\ field is in the interval
    $\left[ \frac{M}{n} - n^c, \frac{M}{n} + n^c \right]$.
\end{corollary}

Setting $c=1$ (so $\delta = \frac{1}{2n}$) gives the following corollary.

\begin{corollary}\label{cor:ave-close-prob-one-over-n}
    For all $t \geq \ln(4M^2)$,
    with probability at least $1-\frac{1}{2n}$,
    within time $t$,
    each agent's \ave\ field is in the interval
    $\left[ \frac{M}{n} - n, \frac{M}{n} + n \right]$.
\end{corollary}

\subsection{\timer}
\label{subsec:timer}

Note that \averaging\ does not actually write the value $\floor{\frac{\M}{\ave} + \frac{1}{2}}$ into the \thecount\ field;
that is the job of the \timer\ protocol, which we now explain.
The leader is guaranteed with high probability to become unique at least by level $\log n$ (Lemma~\ref{lem:uniqueleader}).
However,
since \uniqueID\ likely continues after this point,
although the leader is unique,
when its level increases,
the leader will again generate more bits for its leader code,
updating its value $\M$,
initiating a restart of \averaging.
The problem is that although we can prove that the agents likely
reach level $\log n$ in $O(\log n \log \log n)$ time,
it may take much longer to reach subsequent levels.
Thus,
although the value $\M$ estimated at any level $k \geq \log n$ is large enough for \averaging\ to be correct,
if \averaging\ were to blindly write $\floor{\frac{\M}{\ave} + \frac{1}{2}}$ into \thecount\ each time \ave\ changes,
the output will be disrupted while this restart of \averaging\ converges. 

\begin{algorithm}[ht]
	\floatname{algorithm}{Subprotocol}
	\caption{Timer(rec, sen)}
	\label{protocol:timer-one-way}
	\begin{algorithmic}[100]
        \LeftComment {run phase clock until $\maxPhase=\maxPhaseValue$ is reached}
        \If {$\rec.\isLeader$ and ($\rec.\phase = \sen.\phase$) and ($\rec.\phase < \maxPhase$)}
            \State {$\rec.\phase \gets \rec.\phase + 1$}
        \EndIf
        \If {(not $\rec.\isLeader$) and ($\rec.\phase < \sen.\phase$)}
            \State {$\rec.\phase \gets \sen.\phase$}
        \EndIf
        \State{$\textrm{newCount} \gets \floor{\rec.\M / \rec.\ave + 1/2}$ } \Comment{$\M/\ave$ rounded to the nearest integer}
        \LeftComment{only write output if timer is done and new count is different} 
        \If {($\rec.\phase = \maxPhase$) and ($\rec.\thecount \neq \textrm{newCount}$) and ($\M \geq 3 \cdot \textrm{newCount}^3$)}
            \State {$\rec.\thecount \gets \textrm{newCount}$}
        \EndIf
	\end{algorithmic}
\end{algorithm}

We deal with this problem in the following way.
When the leader restarts \averaging,
it simultaneously restarts \timer,
which is a \emph{phase clock} as described by Angluin et al.~\cite{AAE08}.
\timer\ is so named because we can find $\beta_l < \beta_u$ and $\maxPhase$
such that \maxPhase\ phases of the phase clock will take time between
$\beta_l \ln n$ and $\beta_u \ln n$ with high probability.
So long as $\beta_l \ln n$ is greater than a high-probability upper bound on the running time of \averaging,
the timer likely will not go off (reach the final phase \maxPhase)
until \averaging\ has converged.
It is only once \timer\ has reached phase \maxPhase\ that \thecount\ is written,
and then only if the new calculated size differs from the previous value in \thecount.

There is one additional check done before writing to \thecount:
if newCount $= \floor{\frac{\M}{\ave} + \frac{1}{2}}$,
we must have $\M \geq 3 \cdot \textrm{newCount}^3$ in order to write to \thecount.
In particular,
if $\M \geq \lowerBoundOnM$,
then newCount cannot be $n$ unless $\M \geq 3 \cdot \textrm{newCount}^3$.
This is an optimization to save space.
\averaging\ is only guaranteed to get the correct size $n$ efficiently if $\M \geq \lowerBoundOnM$.
However, when \ave\ is small before convergence
(e.g., 1)
then $\floor{\frac{\M}{\ave} + \frac{1}{2}}$ can be as large as $\M$,
requiring
$\Mconstant \log n$ bits.
But if $\floor{\frac{\M}{\ave} + \frac{1}{2}} = n$
(i.e., is correct)
then this value requires at most
$\log n$ bits.
Since $\M$ could be as large as $3 \cdot n^{\Mconstant}$,
requiring $O(1) + \Mconstant \log n$ bits,
this implies $\thecount$ could be as large as $n^{\Cconstant}$,
requiring $\Cconstant \log n$ bits.

\subsection{\exactCounting\ is fast and correct with high probability}

The following is adapted from~\cite[Corollary 1]{AAE08}.
It relates the number of phases in a phase clock
to upper and lower bounds on the likely time spent getting to that phase.
\opt{sub}{A proof is given in the full version of this paper.}
Our proof appeals entirely to Corollary 1 of \cite{AAE08}
but, unlike \cite{AAE08},
the exact relationship between the constants is given in the lemma statement.


\newcommand{\maxPhaseLemma}{p}

\begin{lemma}[\cite{AAE08}]\label{lem:fastcomputation}
    Let $\beta_l,\epsilon_l,\epsilon_u > 0$,
    and define
    $\maxPhaseLemma =  \max(8 \epsilon_l, 32 \beta_l)$
    and
    $\beta_u = 4 \maxPhaseLemma (\epsilon_u+2)$.
    Let $T_\maxPhaseLemma$ be the time needed for a phase clock with $\geq \maxPhaseLemma$ phases to reach phase $p$.
    Then for all sufficiently large $n$,
    $\Pr[T_\maxPhaseLemma < \beta_l \ln n] < \dfrac{1}{n^{\epsilon_l}}$
    and
    $\Pr[T_\maxPhaseLemma > \beta_u \ln n] < \dfrac{1}{n^{\epsilon_u}}$.
\end{lemma}



\opt{full}{
    \begin{proof}
        Based on~\cite[Corollary 1]{AAE08},
        setting
        (variables of \cite[Corollary 1]{AAE08} on the left,
        and our variables on the right)
        $c = \epsilon_l, d = \beta_l, k = \maxPhaseLemma$, and $a=1/16$,
        then by choosing
        $\maxPhaseLemma =  \max(8 \epsilon_l, 32 \cdot \beta_l)$,
        we have $\Pr[T_\maxPhaseLemma < \beta_l \ln n] < \dfrac{1}{n^{\epsilon_l}}$.
        By Lemma~\ref{lem:epidemic},
        for all $\alpha_u > 0$,
        the epidemic corresponding to each phase $i$ will complete
        (all agents reach phase $i$) in time
        $> \alpha_u \ln n$
        with probability $< 4n^{- \alpha_u / 4 + 1}$.
        Since the time to complete the epidemic is an upper bound
        on the time for the leader to interact with an agent in phase $i$
        (which could happen before the epidemic completes),
        we also have that the phase takes time
        $> \alpha_u \ln n$
        with probability $< 4n^{- \alpha_u / 4 + 1}$.
        By the union bound over all $p$ phases,
        there exists a phase $1 \leq i \leq p$
        taking time
        $> \alpha_u \ln n$
        with probability $< 4 \maxPhaseLemma n^{- \alpha_u / 4 + 1}$.
        Since at least one phase must exceed time
        $\alpha_u \ln n$
        for the sum to exceed
        $\maxPhaseLemma \alpha_u \ln n$,
        $\Pr[T_\maxPhaseLemma >  \maxPhaseLemma \alpha_u \ln n]
        < 4 \maxPhaseLemma n^{-\alpha_u/4+1}$.
        Let $\alpha_u = 4(\epsilon_u+2)$.
        Substituting $\beta_u = \maxPhaseLemma \alpha_u = 4 \maxPhaseLemma (\epsilon_u+2)$
        gives
        $\Pr[T_\maxPhaseLemma >  \beta_u \ln n]
        < 4 \maxPhaseLemma n^{-(4(\epsilon_u+2))/4+1}
        = 4 \maxPhaseLemma n^{-\epsilon_u-1}
        < n^{-\epsilon_u}$,
        which completes the proof.
    \end{proof}
}



The next lemma says that \averaging\ and \timer\ ``happen the way we expect'':
first \averaging\ converges,
before \timer\ ends and records the output of \averaging,
all in $O(\log n)$ time.
When we say ``\averaging\ converges'',
this refers to the \averaging\ protocol running in isolation,
not as part of a larger protocol that might restart it.
That is to say,
it may be that \averaging\ converges,
but \exactCounting\ has not converged,
since \exactCounting\ then restarts \averaging\ and
subsequently changes the \thecount\ field.
Intuitively,
it follows by a simply union bound on the probability that \averaging\ is too slow
(Corollary~\ref{cor:ave-close-prob-one-over-n})
or \timer\ is too fast
(Lemma~\ref{lem:fastcomputation}).
\opt{sub}{A proof is given in the full version of this paper.}

\begin{lemma}\label{lem:averagingconverge}
    For any level $\geq \log n$,
    if it takes $\geq \betaUValue \ln n$ time
    to start the next level,
    with probability $\geq 1-\frac{3}{n}$,
    first \averaging\ converges to the correct output,
    then \timer\ ends and writes $n$ into \thecount,
    in $\leq \betaUValue \ln n$ time.
\end{lemma}

\opt{full}{
    \begin{proof}
        Corollary~\ref{cor:all_length} applied with $\epsilon=1$
        shows that agents codes' length are $\leq 6 \log n$
        with probability $\geq 1-\frac{1}{n}$.
        Lemma~\ref{lem:uniqueleader} shows that after level $\ell$ the leader is unique with probability $\geq 1-\frac{1}{n}$.
        Since $\ell \geq \log n$ the value of $M$ is will be $\geq \lowerBoundOnM$.
        By Lemma~\ref{lem:round}, if \averaging\ converges,
        then $\floor{\frac{\M}{\ave} + \frac{1}{2}}$ exactly $n$.

        By the union bound on Lemma~\ref{lem:fastcomputation} and Corollary~\ref{cor:ave-close-prob-one-over-n},
        with probability $\leq \frac{1}{n} + \frac{1}{n}$ the timer takes more than $\betaUValue \ln n$ time,
        or
        \averaging\ takes more than $\aveTimeConstant \ln n$ time to converge.
        Negating these conditions gives conclusion of the lemma.

        To show that \timer\ does not end until \averaging\ converges,
        we apply Lemma~\ref{lem:fastcomputation} again,
        but using the time lower bound for \timer.
        Letting $\epsilon_l = 1$
        and $\beta_l = \aveTimeConstant$,
        Lemma~\ref{lem:fastcomputation} gives that for
        $p = 32 \beta_l = \maxPhaseValue$,
        with probability $\geq 1 - \frac{1}{n}$,
        \timer\ does not end before $\beta_l \ln n$ time.
        Applying the union bound to this case and the previous two cases then gives probability
        $1 - \frac{1}{n}$ as desired.
    \end{proof}
}

Finally,
we can prove the ``with high probability'' portion of the main theorem.

\begin{theorem} \label{thm:main:whp}
    With probability at least $1-\exactCountingConvergenceTimeProbErr$,
    \exactCounting\
    converges to the correct output
    within $\exactCountingConvergenceTimeConstant \ln n \log \log n$ time,
    and each agent uses at most
    $\stateCountConstantLog + \stateCountExponent \log n$
    bits.
\end{theorem}

\opt{sub}{A formal proof is given in the full version of this paper.}

\begin{proof}[Proof sketch.]
    We sketch the ideas while omitting exact bounds on time and probability.
    Statements below are ``with high probability''.
    Define $k$ to be the unique power of two such that $\log n \leq k < 2 \log n$.
    By Corollary~\ref{cor:quick},
    level $k$ is reached quickly,
    so it suffices to prove fast convergence after the event that $k$ is reached.
    By Lemma~\ref{lem:uniqueleader},
    the leader is unique at level $k$,
    therefore also $2k$ and $4k$,
    and since $k \geq \log n$,
    $\M \geq \lowerBoundOnM$.
    So by Lemma~\ref{lem:averagingconverge},
    \averaging\ will converge within time
    $t = \betaUValue \ln n$.

    We look at three subcases:
    among levels $k, 2k, 4k$,
    one is the earliest among the three where
    $>t$
    time is spent.
    By Lemma~\ref{lem:all_length},
    level $4k$ is not exceeded since codes are unique.
    Since codes are unique at level $4k$,
    at $>t$ (in fact, infinite time)
    will be spent at level $4k$.
    But it could be that the protocol also spends time
    $>t$
    at level $k$ or $2k$.
    Whichever is the first among these three to spend time
    $>t$,
    since the previous spent less time,
    it takes time
    $\leq 2t$
    to reach the first level taking time
    $>t$.
    By Lemma~\ref{lem:averagingconverge},
    \averaging\ converges in time
    $\leq t$,
    and
    by the time \emph{upper} bound of Lemma~\ref{lem:fastcomputation},
    \timer\ reaches phase \maxPhase\
    in time
    $\leq t$
    and records the output of \averaging.

    If we are at level $k$ or $2k$,
    then we might go to a new level.
    If this is guaranteed to happen within $O(\log n)$ time,
    then we would not need the \timer\ protocol.
    We could simply claim that \averaging\
    will converge at the last level reached,
    whether $k$, $2k$, or $4k$.
    The problem that \timer\ solves is that
    \exactCounting\ may reach level $k$ quickly,
    \averaging\ converges quickly,
    yet a small number of duplicate nodes remain,
    say 2.
    It takes $\Omega(n)$ time for them to interact and increase to level $2k$,
    which restarts \averaging.
    By the time \emph{lower} bound of Lemma~\ref{lem:fastcomputation},
    in each of these restarts
    \timer\ will not reach phase \maxPhase\ until \averaging\ reconverges,
    so the \thecount\ field will not be overwritten.
    Thus convergence happened at the \emph{first} level where we spent time
    $>t$,
    even if there are subsequent restarts.
\end{proof}

\opt{full}{

    \begin{proof}
        By Corollary~\ref{cor:quick},
        agents reach level $\log n$ before
        $5 \ln n\log\log n$ time
        with probability at least $1-\frac{5\log\log n}{n}$.
        By Lemma~\ref{lem:uniqueleader},
        the leader is unique at any such level,
        with probability at least $1-\frac{1}{n}$.
        Also $\M = 3 \cdot 2^{3|C|} \geq 3 \cdot 2^{3 \log n} = 3n^2$.
        Hence, if we could stop \uniqueID\ at this point,
        then based on Lemma~\ref{lem:averagingconverge},
        \averaging\ would converge
        in at most $\betaUValue \ln n$ time
        with probability at least $1-\frac{3}{n}$.
        However, there may be duplicate codes,
        so \uniqueID\ possibly continues,
        restarting \averaging\ later.

        Let $t = \betaUValue \ln n$.
        For any $\ell$, define $\round{\ell}$ to be the event that
        \exactCounting\ spends more than time $t$ at level $\ell$.
        Define $k$ to be the unique power of two such that $\log n \leq k < 2\log n$.
        We consider the following disjoint cases that cover all possible outcomes:

        \begin{description}
        \item[\underline{$\round{k}$}:]
            Lemma~\ref{lem:averagingconverge} shows that
            with probability at least $1-\frac{3}{n}$
            \averaging\ converges to the correct output and this output is recorded.
            By the time \emph{upper} bound of Lemma~\ref{lem:fastcomputation},
            with probability at least
            \timer\ reaches phase \maxPhase\
            in time
            $\leq t$
            and records the output of \averaging.
            It remains to show that convergence is likely;
            i.e., this correct value will not be overwritten.
            Again using Lemma~\ref{lem:averagingconverge},
            with probability at most $\frac{3}{n}$,
            at level $2k$ \timer\ ends before \averaging,
            and similarly for error probability $\frac{3}{n}$ at level $4k$.
            By the union bound over these three subcases,
            in this case,
            with probability
            $\geq 1 - \frac{3+3+3}{n} = 1 - \frac{9}{n}$,
            \exactCounting\ converges in time
            $< t$.

        \item[\underline{(not $\round{k}$) and $\round{2k}$}:]
            Similar to above,
            we apply Lemma~\ref{lem:averagingconverge} to level $2k$ to obtain probability
            $\geq 1 - \frac{3}{n}$
            that \averaging\ converges and \timer\ writes $n$ into \thecount\ in time
            $< t$,
            and apply Lemma~\ref{lem:averagingconverge} to level $4k$ to obtain probability
            probability at most $\frac{3}{n}$
            that \timer\ ends too early and disrupts convergence.
            By the union bound over these two subcases,
            in this case,
            with probability
            $\geq 1 - \frac{3+3}{n} = 1 - \frac{6}{n}$,
            \exactCounting\ converges in time
            $< t$.

        \item[\underline{(not $\round{k}$) and (not $\round{2k}$) and $\round{4k}$}:]
            Apply Lemma~\ref{lem:averagingconverge} to level $4k$ to obtain probability
            $\geq 1 - \frac{3}{n}$
            that \averaging\ converges and \timer\ writes $n$ into \thecount\ in time
            $< t$.

        \item[\underline{not $\round{4k}$}:]
            Lemma~\ref{lem:all_length} applied with $\epsilon=2$
            gives probability at most $\frac{1}{n^2}$
            that there are duplicate codes and we reach subsequent levels.
        \end{description}

        Since the four cases are disjoint,
        take the maximum error probability of any of them:
        with probability at least
        $1 - \frac{9}{n}$,
        once at level $k$,
        it takes time $< t$ to converge.

        By the union bound on
        the event that it takes more than time
        $5 \ln n\log\log n$ to reach level $\log n$
        (probability $\leq \frac{5\log\log n}{n}$),
        the event that the leader is not unique
        (probability $\leq \frac{1}{n}$),
        and the event that once at level $k$,
        it takes time $\geq t$ to converge
        (probability $\leq \frac{9}{n}$),
        we obtain that with probability at least
        $1 - \frac{10 + 5 \log \log n}{n}$,
        \exactCounting\ converges in time
        $< 6 \ln n \log \log n$.

        \todo{DD: pull this out into a separate lemma eventually}
        We now prove the memory requirements.
        First, if the maximum level reached is $\ell$,
        then the memory requirements are
        \todo{DD: these are explained a bit elsewhere, but we should re-explain them here.}
        $\ell$ for \C,
        $2 \ell$ for \LC,
        $2 + 3 \ell$ for \M,
        $2 + 3 \ell$ for \ave,
        $\ell$ for \thecount,
        1 for \isLeader,
        and for $\phase$,
        $\log \maxPhase = \log \maxPhaseValue < 12$,
        summing to
        $17 + 10 \ell$.

        There are two disjoint cases:
        $2k \geq 3 \log n$
        and
        $2k < 3 \log n$.
        In the former case,
        applying Lemma~\ref{lem:all_length} with $\epsilon=1$
        gives probability at most $\frac{1}{n}$ of a duplicate code.
        In the latter case,
        $4 \log n \leq 4k < 6 \log n$,
        and applying Lemma~\ref{lem:all_length} with $\epsilon=2$
        gives probability at most $\frac{1}{n^2}$ of a duplicate code.
        In either case the level is less than $\Cconstant \log n$,
        so we set $\ell = \Cconstant \log n$.
        The sum of the bit requirements is then
        $\stateCountConstantLog + \stateCountExponent \log n$ as needed.
        Since the cases are disjoint,
        we take the maximum error probability $\frac{1}{n}$.

        Taking a union bound between this event of ``too much memory''
        and the previous event of ``too much time'',
        the total error probability bound is
        $\exactCountingConvergenceTimeProbErr$
        as required.
        \end{proof}

}

\subsection{\exactCounting\ converges in fast expected time}
\label{subsec:expected-time}

\newcommand{\thmMainExp}{
    The next theorem
    shows a fast expected convergence time,
    and it completes the second portion of the main result,
    Theorem~\ref{thm:main}.
    \opt{sub}{A proof is in the full version of this paper.}

    \begin{theorem} \label{thm:main:exp}
        \exactCounting\ converges in expected time $\expectedTimeConstant \ln n \log \log n$.
    \end{theorem}
}

    Most of the technical difficulty of our analysis is captured by the ``with high probability'' results stated already.
    \exactCounting\ is also stabilizing,
    meaning that with probability 1 it gets to a correct configuration that is \emph{stable} (the output cannot change).
    Probability 1 correctness is required for the expected correct convergence time to be finite,
    and indeed it asymptotically matches the high probability convergence time of $O(\log n \log \log n)$.
    However, the protocol takes longer to stabilize,
    up to $O(n)$ time,
    since it does not stabilize until \uniqueID\ stabilizes.\footnote{
        Prior to that,
        it is possible,
        with low probability,
        after $n$ is written into each agent's \thecount\ field,
        for a subsequent restart of \averaging\ to write incorrect values,
        if the corresponding restart of
        \timer\ completes too quickly.
        So no configuration is stable until \uniqueID\ converges and triggers the final restart.
    }

    \thmMainExp

\opt{full}{
    First we establish some other claims necessary to prove Theorem~\ref{thm:main:exp}.
    Recall that a protocol \emph{stabilizes} if it converges to the correct output with probability 1.

    \begin{lemma}
    \exactCounting\ stabilizes to the correct population size.
    \end{lemma}

    \begin{proof}
        Since each agent generates code bits uniformly at random,
        any pair of agents has probability 0 to generate the same infinite sequence of bits.
        So with probability 1 all agents eventually have unique codes,
        and \uniqueID\ stabilizes.
        We now show that implies \electLeader\ stabilizes.

        Since there are $n$ agents, \uniqueID\ cannot terminate until at least level $\log n$.
        So when \uniqueID\ terminates,
        $|\C| \geq \log n$,
        so $\M = 3 \cdot 2^{3|\C|} \geq \lowerBoundOnM$.
        Note that \averaging\ also has an equivalence between converging and stabilizing:
        one all agents \ave\ fields are within a certain interval, they cannot leave that interval.
        So by Lemma~\ref{lem:round},
        \averaging, if it stabilizes,
        will stabilize to values of \ave\ such that $\floor{\frac{\M}{\ave} + \frac{1}{2}} = n$.
        We claim that \averaging\ stabilizes with probability 1,
        which is shown below.
        Furthermore,
        \timer\ reaches \maxPhase\ with probability 1,
        since the only way to avoid incrementing the phase of an agent is forever to avoid any interaction between it and an agent at the next phase,
        which happens with probability 0.
        This implies that with probability 1,
        the correct population size is eventually written into the field {\thecount}.

        It remains to show the claim that \averaging\ stabilizes with probability 1.
        Define the potential function $\Phi$ for any configuration $\vec{c}$
        by $\Phi(\vec{c}) = \sum_{a} |a.\ave - \M/n|$,
        where the sum is over each agent $a$ in the population.
        The \averaging\ protocol stabilizes by the time $\Phi$ reaches its minimum value,\footnote{
            If $\M \gg 2n^2$,
            then \averaging\ can stabilize prior to this time,
            since the \ave\ values do not have to reach their final convergent values for
            the output function $\floor{\frac{\M}{\ave} + \frac{1}{2}}$ to converge,
            by Lemma~\ref{lem:round}.
        }
        which is either $n$ or $0$ depending on whether $n$ divides $\M$,
        when all agents have ave $= \floor{\M/n}$ or $\ceil{\M/n}$.
        We claim that $\Phi$ is nonincreasing with each transition of \averaging.
        When two agents meet,
        there are two cases:
        1) both of their \ave\ fields are $\geq \ceil{\M/n}$,
        or both are $\leq \floor{\M/n}$,
        and
        2) one of their \ave\ fields is $<$ (resp., $\leq$) $\floor{\M/n}$,
        and the other is $\geq$ (resp., $>$) $\ceil{\M/n}$ .
        Taking the average of their \ave\ fields,
        in case (1) does not change $\Phi$,
        and in case (2) decreases $\Phi$,
        so $\Phi$ is nonincreasing.

        It remains to show that $\Phi$ will reach its minimum value with probability 1.
        If the protocol has not converged,
        then there must be some agent with an \ave\ field not equal to
        either $\floor{\M/n}$ or $\ceil{\M/n}$.
        But since the population-wide sum of the \ave\ values is always $\M$,
        this implies that case (2) holds for some pair of agents.
        With probability 1, such a pair of agents must eventually meet,
        decreasing $\Phi$.
        So with probability 1,
        $\Phi$ eventually reaches its minimum value.
    \end{proof}

    \uniqueID\ stabilizes when all agents have a unique code since,
    by inspection of the \uniqueID\ protocol,
    this implies that the codes no longer can change.
    The next lemma shows that this happens at most linear time.

    \def\half{0.5}
    \def\eps{0.01}
    \FPeval{\constantA}{clip(\half+\eps)}
    \FPeval{\constantB}{clip(2*\constantA)}
    \FPeval{\constantC}{clip(\constantB+\eps)}
    \FPeval{\constantD}{clip(\constantC+\eps)}

    \begin{lemma} \label{lem:unique-ID-converge}
    \uniqueID\ stabilizes in expected time at most $\constantC n$.
    \end{lemma}

    \begin{proof}
        By Lemma~\ref{lem:epidemic}, once one agent reaches a new level,
        the expected time for all agents to reach that level is at most
        $\epidemicExpectedTimeConstant \ln n$.
        Once all agents are at the same level and there is at least one pair of duplicate codes
        (i.e., \uniqueID\ has not yet stabilized),
        the probability that the next interaction is a pair of agents with the same code is
        at least $1 / \binom{n}{2} > 2/n^2$,
        so the expected number of interactions for these agents to meet and start a new level is
        at most $n^2/2$,
        so expected time $n/2$.
        Thus, once there is one agent at a level,
        the expected time to get an agent at the next level
        (assuming \uniqueID\ does not stabilize at the current level)
        is at most
        $\epidemicExpectedTimeConstant \ln n + n/2 < \constantA n$.

        By Lemma~\ref{lem:quick},
        setting $\alpha = 1$,
    	with probability at least
        $1 - \frac{5 \log \log n}{n}$,
        in $5 \ln n \log \log n$ time all agents reach a level $k$ such that $\log n \leq k < 2 \log n$.
    	Once there,
    	if duplicate codes remain,
    	it takes expected time $< \constantA n$ to reach level $2k$ by the above argument.
    	If duplicate codes still remain,
    	it similarly takes expected time $< \constantA n$ to reach level at least $4k$.
    	By Lemma~\ref{lem:all_length},
    	for any $j \geq 0$
    	(letting $\epsilon = j-4$ in Lemma~\ref{lem:all_length}),
        duplicate codes remain at level $j \log n = (4+\epsilon) \log n$ with probability at most $n^{-(j-4)}$.

        Each new level after that takes $< \constantA n$ expected time.
        Thus the total expected time is at most
        (letting $j$ above be $2^i$ below):
        \begin{eqnarray*}
            &&
            \underbrace{5 \ln n \log \log n}_{\text{time to reach level $k \geq \log n$}}
            +
            \underbrace{\constantA n}_{\text{time to reach level $2k \geq 2 \log n$}}
            +
            \underbrace{\constantA n}_{\text{time to reach level $4k \geq 4 \log n$}}
            \\&&
            +
            \underbrace{\sum_{i=3}^\infty \constantA n \cdot \Pr[\text{level $2^i \log n$ has duplicate codes}]}_{\text{contribution of levels $\geq 8 \log n$}}
            \\ &\leq&
            5 \ln n \log \log n + \constantB n +
            \constantA n \sum_{i=3}^\infty n^{-(2^i - 4)} \ \ \ \ \ \ \ \ \ \ \ \ \ \ \text{Lemma~\ref{lem:all_length}}
            \\ &<&
            5 \ln n \log \log n + \constantB n +
            \constantA n \sum_{i=1}^\infty n^{-2^i + 2}
            \\ &=&
            5 \ln n \log \log n + \constantB n + 1 +
            \constantA n \sum_{i=2}^\infty n^{-2^i + 2}
            \\ &<&
            5 \ln n \log \log n + \constantB n + 1 +
            \constantA n \sum_{i=2}^\infty n^{-i}
            \\ &<&
            5 \ln n \log \log n + \constantB n + 1 +
            \constantA n \sum_{i=1}^\infty n^{-i}
            \\ &=&
            5 \ln n \log \log n + \constantB n + 1 +
            \constantA n \left( \frac{1}{1-n^{-1}} - 1 \right) \ \ \ \ \text{geometric series}
            \\ &=&
            5 \ln n \log \log n + \constantB n + 1 +
            \constantA n \frac{1}{n-1}
            < 
            \constantC n.   \qedhere
        \end{eqnarray*}
    \end{proof}

    \begin{proof}[Proof of Theorem~\ref{thm:main:exp}]
        By Theorem~\ref{thm:main:whp},
        \exactCounting\ converges to the correct answer in
        time $\exactCountingConvergenceTimeConstant \ln n \log \log n$
        with probability at least
        $1 - \exactCountingConvergenceTimeProbErr$.

        By Lemma~\ref{lem:unique-ID-converge},
        \uniqueID\ converges in expected time at most $\constantC n$.
        Once it has converged,
        it takes expected time
        $O(\log n)$
        for \averaging\ to converge and \timer\ to write the correct output if it has not already been written.
        The sum of these times is at most $\constantD n$ for sufficiently large $n$.
        We can bound the expected time as
        \begin{eqnarray*}
        &&
        \Pr[\text{convergence in time } \leq \exactCountingConvergenceTimeConstant \ln n \log \log n]
        \cdot \exactCountingConvergenceTimeConstant \ln n \log \log n +
        \\&&
        \Pr[\text{convergence in time } > \exactCountingConvergenceTimeConstant \ln n \log \log n]
        \cdot \constantD n
        \\&=&
        \left( 1 - \exactCountingConvergenceTimeProbErr \right)
        \cdot \exactCountingConvergenceTimeConstant \ln n \log \log n
        +
        \exactCountingConvergenceTimeProbErr \cdot \constantD n
        \\&<&
        \exactCountingConvergenceTimeConstant \ln n \log \log n
        +
        \constantD (\exactCountingConvergenceTimeProbErrNumerator)
        \\&<&
        \expectedTimeConstant \ln n \log \log n. \qedhere
        \end{eqnarray*}
    \end{proof}

}

\opt{full}{
    \subsection{Increasing time to minimize state complexity}
    \label{subsec:minimize-states}

    \exactCounting\ generalizes straightforwardly to trade off time and memory:
    by adjusting the rate at which the code length grows,
    the convergence time $t(n) = O(f(n) \log n)$,
    where $f(n)$ is the number of stages required for the code length to reach $\log n$.
    The minimum state complexity is achieved
    when the code length increments by 1 each stage,
    so that $f(n) = \log n$
    and $t(n) = \log^2 n$.
    In this case,
    a straightforward adaptation of
    Lemma~\ref{lem:all_length},
    letting $\epsilon=1$,
    indicates that with probability at least $1 - \frac{1}{n}$,
    all codes are unique by level $\CconstantMin \log n$.
    Carrying through the string length and integer bounds from the main argument
    gives state complexity $O(n^{\stateCountExponentMin})$
    for the full protocol
    and $O(n^{\LCsubprotocolStateCountConstantMin})$
    for just the leader election.
}

\subsection{Experiments}
\label{sec:all_experiments}


Simulations for the \exactCounting\ protocol
are shown in Figure \ref{fig:counting-plot}.

\begin{figure}[th]
	\centering
	\includegraphics[width=\textwidth, draft=false]{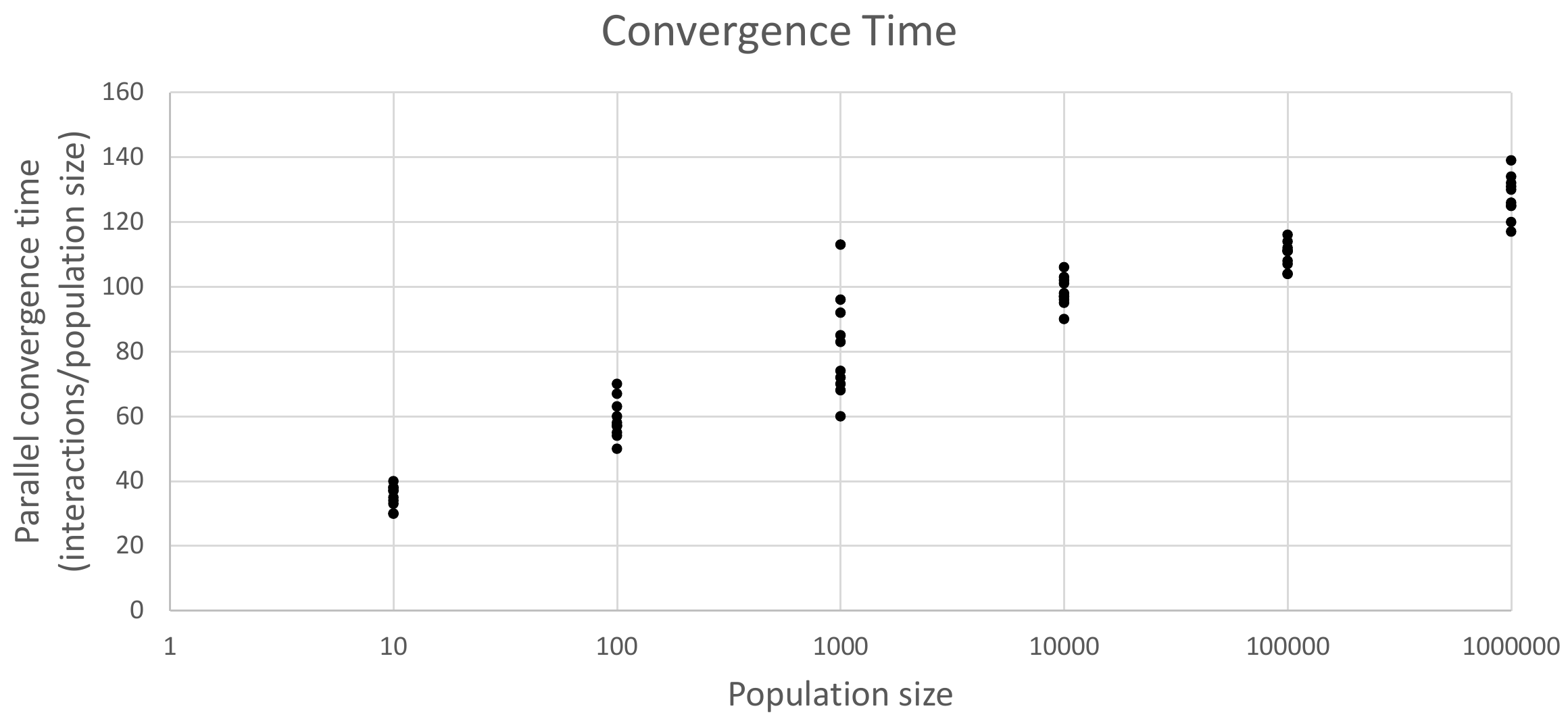} 
	\caption{Simulated convergence time of \exactCounting.
	The dots indicate the convergence time of individual experiments.
	The population size axis is logarithmic,
	so exactly $c \log_{10} n$ time complexity would correspond to a line of slope $c$.
	Since $\log \log n$ is ``effectively constant'' ($< 5$) for the values of $n$ shown,
	we similarly expect the plot to appear roughly linear.}
	\label{fig:counting-plot}
\end{figure}

\section{Conclusion}
\label{sec:conclusions}

We have shown a uniform population protocol computing the exact population size using
$O(\log n)$ bits memory
(i.e., $\mathrm{poly}(n)$ states)
and
$O(\log n \log \log n)$ time.
By removing the \averaging\ and \timer\ subprotocols,
the remainder is a uniform protocol electing a leader in
$O(\log n \log \log n)$ time
and
$18 \log n$ bits of memory (for $\C$ and $\LC$).

Some interesting questions are open.
Is there a \emph{uniform} polylogarithmic time population protocol, correct with high probability, for the problem of...

\begin{enumerate}

    \item \label{ques:leaderElectionTerminating}
    leader election, which is \emph{terminating}?

    \item \label{ques:approxSizeTerminating}
    constant-factor approximate size estimation, which is \emph{terminating}?

    \item
    exact size computation, which is \emph{terminating}?

    \item \label{ques:leaderElectionPolylogState}
    leader election, which is \emph{$\mathrm{polylog}(n)$ state-bounded}?

    \item \label{ques:sizeApproxPolylogState}
    constant-factor approximate size estimation, which is \emph{$\mathrm{polylog}(n)$ state-bounded}?

    \item \label{ques:sizeLinearState}
    exact size computation, which is \emph{$O(n)$ state-bounded}?

\end{enumerate}

For Question~\ref{ques:sizeLinearState},
the trivial lower bound is $n$,
but the \averaging\ protocol seems intuitively to require $\Omega(n^2)$ states.
It would be interesting to prove a $\Omega(n^2)$ lower bound,
either for ``any scheme based on averaging'' (suitably formalized),
or more generally for any sublinear-time size counting protocol.
Since \averaging\ requires $O(n^2)$ states,
if the initial configuration has a leader and a constant-factor approximation of $n$,
this means that solutions to questions
\ref{ques:leaderElectionTerminating},
\ref{ques:approxSizeTerminating},
\ref{ques:leaderElectionPolylogState},
and
\ref{ques:sizeApproxPolylogState}
would immediately imply a
$O(n^2)$ state,
$\mathrm{polylog}(n)$ time protocol for
exact population size computation.

Since many problems such as leader election require only an estimate on the population size,
not an exact value,
a protocol answering questions~\ref{ques:approxSizeTerminating} and~\ref{ques:sizeApproxPolylogState}
is an important goal.
Alistarh, Aspnes, Eisenstat, Gelashvili, Rivest~\cite{alistarh2017time}
have shown a uniform protocol (converging, but not terminating)
using only $\mathrm{polylog}(n)$ states
that in $O(\log n)$ expected time
can get an estimate $n'$ within a polynomial
(but not constant)
factor of the true size $n$:
with high probability $1/2 \log n \leq \log n' \leq 9 \log n$
i.e., $\sqrt{n} \leq n' \leq n^9$.

\newpage

\bibliographystyle{plain}
\bibliography{My-Library} 




\appendix

\end{document}